\documentclass[11pt, a4paper]{article}
\usepackage{amsmath, amsfonts, amssymb, amsthm}
\usepackage[top=1in, bottom=1in, left=1in, right=1in]{geometry}
\usepackage[T1]{fontenc}
\usepackage{lmodern}
\usepackage{longtable,array,enumitem,cite}
\usepackage[hidelinks,hyperfootnotes=false,hypertexnames=false]{hyperref}
\usepackage{booktabs,float}
\usepackage{microtype}

\setlist[enumerate]{itemsep=2pt,topsep=3pt}
\newcommand{\coef}[1]{\texttt{#1}}
\allowdisplaybreaks[1]

\newcolumntype{V}{!{\vrule width 0.5pt}}
\newcommand{\twocolsetup}{%
  \renewcommand{\arraystretch}{1.0}%
  \setlength{\aboverulesep}{0pt}%
  \setlength{\belowrulesep}{0pt}%
  \setlength{\extrarowheight}{1.5pt}%
  \setlength{\belowcaptionskip}{6pt}}

\DeclareMathOperator{\lcm}{lcm}
\newtheorem{theorem}{Theorem}[section]
\newtheorem{lemma}[theorem]{Lemma}
\newtheorem{corollary}[theorem]{Corollary}
\newtheorem{proposition}[theorem]{Proposition}
\theoremstyle{definition}
\newtheorem{example}[theorem]{Example}
\newtheorem{remark}[theorem]{Remark}
\newtheorem{definition}[theorem]{Definition}

\begin{document}
\title{New Nonbinary Quantum Codes from Quasi-Twisted Codes with Scalar and Matrix Twists}
\author{Tushar Bag$^{1,2}$, Hai Q. Dinh$^3$, Daniel Panario$^4$}
\date{
\small{
1. Quantum Technology Institute,  SRM University AP, Amaravati 522240, Andhra Pradesh, India \\
2. Department of Mathematics,  SRM University AP, Amaravati 522240, Andhra Pradesh, India \\
3. Kent State University, Warren, OH 44483, USA\\
4. School of Mathematics and Statistics, Carleton University, Ottawa, Ontario K1S 5B6, Canada.
}
}

\begingroup
\renewcommand\thefootnote{}
\footnotetext{
Email: tusharbag2011@gmail.com (T. Bag) [corresponding author], 
hdinh@kent.edu (H. Q. Dinh), 
daniel@math.carleton.ca (D. Panario)
}
\endgroup

\maketitle

\begin{abstract}
We construct quantum codes over nonbinary finite fields from one-generator index-$2$ quasi-twisted
codes and from codes invariant under a $2\times2$ matrix twist, via the
symplectic construction. For one-generator quasi-twisted codes, we give direct
proofs of the dimension formulas and of matrix and polynomial tests for
symplectic self-orthogonality. When the same scalar twist $\lambda$ acts on both blocks, we show
that $\lambda^2\ne1$ gives only quantum distance $1$. Thus, in this
construction, distance greater than $1$ requires $\lambda=\pm1$,
corresponding to quasi-cyclic and quasi-negacyclic codes.
We also study determinant-one $2\times2$ matrix twists.
For twists given by companion matrices, we describe the resulting
codes and their symplectic duals by polynomials over the base field.
When the associated quadratic is irreducible over $\mathbb F_q$,
these codes correspond to Hermitian constacyclic codes over the
quadratic extension. With a rank criterion for the minimum distance and standard propagation rules, our explicit constructions present quantum codes that improve the recorded minimum-distance lower bounds for $164$ distinct entries in Grassl's tables of quantum codes over nonbinary finite fields. 
\end{abstract}

\noindent\textbf{Keywords:} quasi-twisted codes;  symplectic self-orthogonal codes;  nonbinary quantum codes.\\
\noindent\textbf{Mathematics Subject Classification (2020):} 94B05, 94B15, 94B60, 81P70.

\section{Introduction}
Let $\mathbb F_q$ be the finite field with $q$ elements. We study quasi-twisted codes over $\mathbb F_q$. Quasi-twisted codes generalize quasi-cyclic codes. A $\lambda$-quasi-twisted code is closed under a constacyclic shift by its index, where $\lambda\in\mathbb F_q^*$. We write QC for quasi-cyclic. The choice $\lambda=1$ gives a QC code. Aydin et al.\ \cite{Nuh1} described one-generator quasi-twisted codes and used their structure to construct linear codes. Jia \cite{Jia2012} studied their components, including the repeated-root case. For QC codes, Conan and S\'eguin \cite{Con} studied structure and enumeration, Lally and Fitzpatrick \cite{La} used polynomial modules, and S\'eguin \cite{Se} studied one-generator codes. Ling and Sol\'e developed the finite-field and generator theory in \cite{L1,L3}. Ling et al.\ \cite{L4} treated repeated roots.

Shor \cite{S95}, Calderbank and Shor \cite{cal1996}, and Steane \cite{St} introduced quantum error correction. The relation with classical codes gives the symplectic construction \cite{CRSS98,A,Ket}. Several authors used QC codes to construct quantum codes \cite{Gal,Q3,Q2,Q4}. Guan et al.\ \cite{Q1} gave necessary and sufficient conditions for symplectic self-orthogonality of QC codes of even index, including the one-generator case considered here. Duals of one-generator QC codes and their quantum applications were studied in \cite{QC,iitr,Biswas2023}. Dastbasteh and Shivji \cite{Dastbasteh2025} studied the related additive cyclic codes. Abdukhalikov et al.\ \cite{ADL2026} gave a general description of index-$2$ QC codes and their duals.

Quasi-twisted codes have also been used to construct quantum codes. Lv et al.\ \cite{LvQT2020} used Hermitian self-orthogonality. Ezerman et al.\ \cite{Ezerman2025} studied nearly self-orthogonal quasi-twisted codes and Construction X. They also treated reciprocal mixed twists and their relation to the Calderbank--Shor--Steane construction. Abdukhalikov and Verma \cite{AbVerma2026} determined duals of index-$2$ quasi-twisted codes, gave self-orthogonality criteria for $\lambda=\pm1$, and related scalar quasi-twisted codes to additive constacyclic codes over extension fields. Zhang et al.\ \cite{Zhang2026} studied symplectic duals of quasi-twisted codes and constructed ternary quantum codes.

In this paper, we first consider the case in which the same scalar $\lambda$ acts on both symplectic blocks and $\lambda^2\ne1$. The shift then changes the symplectic product by a term supported at the boundary. We use this term to prove Theorem~\ref{nogo}: every nonzero symplectic self-orthogonal index-$2$ $\lambda$-quasi-twisted code has the form $C=\{(\alpha u\mid\beta u):u\in D\},$
where $D$ is a $\lambda$-constacyclic code and $(\alpha,\beta)\ne(0,0)$. The associated quantum code has distance $1$. Thus the same-scalar construction can give distance greater than $1$ only when $\lambda^2=1$. Earlier treatments already single out $\lambda=\pm1$ in this setting; see \cite{Ezerman2025} and \cite{AbVerma2026}. Theorem~\ref{nogo} describes all such codes explicitly when $\lambda^2\ne1$. We therefore restrict the scalar search to $\lambda=\pm1$.

For one-generator quasi-twisted codes, we give direct proofs of the standard annihilator and dimension formulas because they hold without assuming that $x^n-\lambda$ is square-free. We then write the self-orthogonality tests in matrix and polynomial form.

We first give six one-generator quasi-cyclic constructions, O1--O6. We then construct new quantum codes T1--T4 from quasi-negacyclic codes and N1--N44 from quasi-cyclic codes.

Finally, we construct the new quantum codes H1--H13 from Hermitian self-orthogonal constacyclic codes over $\mathbb F_{q^2}$ with a twist of norm $1$. This is an established method; see Kai et al.\ \cite{Kai2014}, Chen et al.\ \cite{Chen2015}, and Lv et al.\ \cite{LvQT2020}. We write multiplication by the extension-field twist as a $2\times2$ twist matrix over $\mathbb F_q$. Proposition~\ref{matrix-ideal-dual} and Theorem~\ref{hermitian-dual} give the ideal dual and the Hermitian correspondence in these coordinates. This description keeps the generators and symplectic tests over the base field. The scalar and matrix-twist constructions use polynomial arithmetic over $\mathbb F_q$, and the latter also uses the norm from $\mathbb F_{q^2}$ to $\mathbb F_q$.

The quantum codes constructed here, together with standard propagation rules, improve the recorded minimum-distance lower bounds for $164$ distinct entries in Grassl's tables of quantum codes over nonbinary finite fields \cite{gra}. Tables~\ref{tab:negacyclic}, \ref{tab:qc-direct}, and~\ref{tab:qt-direct} give 61 explicit constructions over $\mathbb F_3$, $\mathbb F_5$, and $\mathbb F_7$. Of these, 52 provide, directly or through the propagation rules of Section~\ref{sec:propagation}, improved distance bounds for 163 distinct entries. Table~\ref{tab:dominated} identifies the other nine constructions; we also give their generators in the appendices. The quasi-cyclic code O6 improves one additional entry. We count each triple $(q,n,k)$ once, taking the strongest distance bound. 
Section~\ref{sec:preliminaries} gives the notation and quantum construction. Section~\ref{sec:one-generator} treats one-generator quasi-twisted codes and the scalar restriction. Section~\ref{sec:constructions} presents the matrix description, code parameters, a rank criterion for the quantum minimum distance, and propagation. Section~\ref{sec:conclusion} concludes the paper. The appendices give the generator coefficients.

\section{Preliminaries}\label{sec:preliminaries}

Let $\mathbb F_q$ be the finite field with $q=p^e$ elements, where $p$ is
prime and $e\ge1$. A linear code $C$ of length $m$ over $\mathbb F_q$ is a
subspace of $\mathbb F_q^m$. For a vector
$a=(a_0,\ldots,a_{m-1})\in\mathbb F_q^m$, let
$w_H(a)=|\{i:a_i\ne0\}|$ denote its Hamming weight. The minimum Hamming
distance of a nonzero linear code $C\subseteq\mathbb F_q^m$ is
$d_H(C)=\min\{w_H(c):0\ne c\in C\}$. We write $[m,k,d]_q$ for a linear
code of length $m$, dimension $k$, and minimum Hamming distance $d$.

For vectors $a,b\in\mathbb F_q^n$, we write the Euclidean inner product as
$\langle a,b\rangle_e=a\cdot b=\sum_{i=0}^{n-1}a_ib_i$.
We write vectors of $\mathbb F_q^{2n}$ in block form $U=(u\mid v)$, where
$u,v\in\mathbb F_q^n$. The symplectic inner product of $U=(u\mid v)$ and
$V=(u'\mid v')$ in $\mathbb F_q^{2n}$ is
\[
\langle U,V\rangle_s
=\langle u,v'\rangle_e-\langle v,u'\rangle_e
=U\Omega V^t,
\qquad
\Omega=\begin{pmatrix}O_n&I_n\\-I_n&O_n\end{pmatrix},
\]
where $I_n$ and $O_n$ denote the identity and zero matrices of order $n$,
respectively. For a linear code $C\subseteq\mathbb F_q^{2n}$, we define its
symplectic dual as
\[
C^{\perp_s}=\{U\in\mathbb F_q^{2n}:\langle U,V\rangle_s=0
\text{ for all }V\in C\}.
\]

The matrix $\Omega$ is invertible. Hence the symplectic form is
nondegenerate, and elementary linear algebra gives
\[
\dim C^{\perp_s}=2n-\dim C,\qquad (C^{\perp_s})^{\perp_s}=C.
\]
A code \(C\) is symplectic self-orthogonal if \(C\subseteq C^{\perp_s}\), and symplectic dual-containing if \(C^{\perp_s}\subseteq C\). For $c=(u\mid v)\in\mathbb F_q^{2n}$, where $u,v\in\mathbb F_q^n$, we define the symplectic weight by $w_s(c)=|\{i:(u_i,v_i)\ne(0,0)\}|$.

\subsection{Quantum codes from symplectic self-orthogonal codes}

We write $[[n,k,d]]_q$ for a quantum code of length $n$, dimension $k$,
and minimum distance $d$. For a code $C\subseteq\mathbb F_q^{2n}$, we define
its minimum symplectic distance by $d_s(C)=\min\{w_s(c):0\ne c\in C\},$ with the convention $d_s(\{0\})=\infty$.

We use the following standard result.

\begin{theorem}\cite{Ket}\label{Q}
Let $C\subseteq\mathbb F_q^{2n}$ be an $\mathbb F_q$-linear code of
dimension $r$ such that $C\subseteq C^{\perp_s}$. If $r<n$, then there is a
quantum code with parameters $[[n,n-r,d]]_q$, where
\[
d=\min\{w_s(c):c\in C^{\perp_s}\setminus C\}.
\]
If $r=n$, then $C=C^{\perp_s}$, and there is a quantum code with parameters
$[[n,0,d]]_q$, where $d=d_s(C)$.
\end{theorem}

We call this the quantum code associated with $C$. If $r<n$, then
$C^{\perp_s}$ is the disjoint union of $C$ and $C^{\perp_s}\setminus C$, so
\[
d_s\bigl(C^{\perp_s}\bigr)=\min\{d_s(C),\,d\}.
\]

\begin{definition}
Let $C\subseteq C^{\perp_s}$, and let $d$ be the distance of the associated
quantum code. The quantum code is \emph{pure} if $d_s(C)\ge d$, and
\emph{impure} if $d_s(C)<d$.
\end{definition}

If the code is pure, then $d_s(C^{\perp_s})=d$. If it is impure, then
$d_s(C^{\perp_s})=d_s(C)<d$.

\begin{remark}
If $r=n$, then $C=C^{\perp_s}$, and the quantum code has parameters
$[[n,0,d]]_q$ with $d=d_s(C)$. Since the set $C^{\perp_s}\setminus C$ is
empty, the inequality $d_s(C)\ge d$ holds with equality, so the code is
pure.
\end{remark}

We also use the dual-containing form of the same construction.

\begin{corollary}\label{dc}
Let $C\subseteq\mathbb F_q^{2n}$ have dimension $r$ and satisfy
$C^{\perp_s}\subseteq C$. Then $r\ge n$. If $r>n$, there is a quantum code
with parameters $[[n,r-n,d]]_q$, where
\[
d=\min\{w_s(c):c\in C\setminus C^{\perp_s}\}.
\]
If $r=n$, then $C=C^{\perp_s}$, and the case $r=n$ of
Theorem~\ref{Q} applies.
\end{corollary}

\begin{proof}
Put $D=C^{\perp_s}$. Since the symplectic form is nondegenerate,
$D^{\perp_s}=C$ and $\dim D=2n-r$. The inclusion $D\subseteq C$ gives
$2n-r\le r$, and it also gives $D\subseteq D^{\perp_s}$. Applying
Theorem~\ref{Q} to $D$ gives the result.
\end{proof}

\subsection{Constacyclic shifts and polynomial coordinates}\label{sec:shifts}

Let $\lambda\in\mathbb F_q^*=\mathbb F_q\setminus\{0\}$, and put
$R_\lambda=\mathbb F_q[x]/\langle x^n-\lambda\rangle$. Since
$\mathbb F_q[x]$ is a principal ideal domain, $R_\lambda$ is a principal
ideal ring. For $a=(a_0,\ldots,a_{n-1})\in\mathbb F_q^n$, the
$\lambda$-constacyclic shift is
\[
T_\lambda(a)=(\lambda a_{n-1},a_0,\ldots,a_{n-2}).
\]
We use the same symbol for the matrix that performs this shift on a column
vector:
\[
T_\lambda=\begin{pmatrix}
0&0&\cdots&0&\lambda\\
1&0&\cdots&0&0\\
0&1&\ddots&\vdots&\vdots\\
\vdots&\ddots&\ddots&0&0\\
0&\cdots&0&1&0
\end{pmatrix},\qquad
T_\lambda\begin{pmatrix}a_0\\a_1\\\vdots\\a_{n-1}\end{pmatrix}
=\begin{pmatrix}\lambda a_{n-1}\\a_0\\\vdots\\a_{n-2}\end{pmatrix}.
\]
A direct multiplication gives $T_\lambda^tT_{\lambda^{-1}}=I_n$, that is,
$T_\lambda^{-t}=T_{\lambda^{-1}}$.

We identify $a$ with the polynomial
$a(x)=\sum_{i=0}^{n-1}a_ix^i\in R_\lambda$. Since $x^n=\lambda$ in
$R_\lambda$,
\[
x\,a(x)\equiv\lambda a_{n-1}+a_0x+\cdots+a_{n-2}x^{n-1}\pmod{x^n-\lambda},
\]
whose coefficient vector is $T_\lambda(a)$. Thus, under the
identification $a\leftrightarrow a(x)$, the shift $T_\lambda$ corresponds
to multiplication by $x$ in $R_\lambda$. Consequently, a linear code of
length $n$ is $\lambda$-constacyclic, that is, invariant under
$T_\lambda$, if and only if it is an ideal of $R_\lambda$. For $\lambda=1$
these are the cyclic codes.

We now pass to codes of length $ln$, where $l\ge1$. We write a vector of
length $ln$ as $n$ consecutive $l$-tuples,
\[
c=(c_{0,0},\ldots,c_{0,l-1},\;c_{1,0},\ldots,c_{1,l-1},\;\ldots,\;c_{n-1,0},\ldots,c_{n-1,l-1}),
\]
and place these $l$-tuples as the rows of an $n\times l$ matrix:
\[
X=\begin{pmatrix}
c_{0,0}&\cdots&c_{0,l-1}\\
c_{1,0}&\cdots&c_{1,l-1}\\
\vdots&&\vdots\\
c_{n-1,0}&\cdots&c_{n-1,l-1}
\end{pmatrix}.
\]

\begin{definition}
A linear code $C\subseteq\mathbb F_q^{ln}$ is
\emph{$\lambda$-quasi-twisted} \emph{of index $l$} if, for
every $c\in C$,
\[
(\lambda c_{n-1,0},\ldots,\lambda c_{n-1,l-1},\;c_{0,0},\ldots,c_{0,l-1},\;\ldots,\;c_{n-2,0},\ldots,c_{n-2,l-1})\in C .
\]
Equivalently, the matrix $T_\lambda X$ represents a word in $C$ whenever
$X$ represents a word in $C$, where
\[
T_\lambda X=\begin{pmatrix}
\lambda c_{n-1,0}&\cdots&\lambda c_{n-1,l-1}\\
c_{0,0}&\cdots&c_{0,l-1}\\
\vdots&&\vdots\\
c_{n-2,0}&\cdots&c_{n-2,l-1}
\end{pmatrix}.
\]
For $\lambda=1$, $C$ is a quasi-cyclic (QC) code of index $l$.
\end{definition}

Each $l$-tuple moves one place forward, and the last $l$-tuple
comes to the front multiplied by $\lambda$. 

To describe quasi-twisted codes by polynomials, we look at the columns of $X$.
Column $j$ is the vector
$c_j=(c_{0,j},c_{1,j},\ldots,c_{n-1,j})\in\mathbb F_q^n$, for
$0\le j\le l-1$. Since $T_\lambda$ multiplies $X$ from the left, it
shifts each column separately:
\[
T_\lambda X=\bigl(T_\lambda c_0\;\big|\;T_\lambda c_1\;\big|\;\cdots\;\big|\;T_\lambda c_{l-1}\bigr).
\]
So if we list the coordinates column after column, as
$(c_0\mid c_1\mid\cdots\mid c_{l-1})$, the shift applies $T_\lambda$ to
each of the $l$ blocks. This reordering of coordinates is the column
permutation in \cite{Nuh1}. Identifying each block $c_j$ with
$c_j(x)=\sum_{i=0}^{n-1}c_{i,j}x^i\in R_\lambda$, the shift becomes
\[
\bigl(c_0(x),\ldots,c_{l-1}(x)\bigr)\longmapsto\bigl(xc_0(x),\ldots,xc_{l-1}(x)\bigr).
\]
Since $C$ is $\mathbb F_q$-linear, closure under multiplication by $x$
gives closure under multiplication by every element of $R_\lambda$. Hence
a $\lambda$-quasi-twisted code of index $l$ is an $R_\lambda$-submodule of
$R_\lambda^l$, and conversely \cite{Nuh1}.

In this paper we work with index $l=2$. We write $u_i=c_{i,0}$ and
$v_i=c_{i,1}$, so that the rows of $X$ are the pairs $(u_i,v_i)$ and the
columns are $u$ and $v$. We use the block ordering
\[
(u\mid v)=(u_0,\ldots,u_{n-1}\mid v_0,\ldots,v_{n-1})\longleftrightarrow(u(x),v(x))\in R_\lambda^2,
\]
in which the shift is
\[
S_\lambda(u\mid v)=(T_\lambda u\mid T_\lambda v),
\]
corresponding to $(xu(x),xv(x))$. Thus $S_\lambda$ is the quasi-twisted shift in
block coordinates, and an index-$2$ $\lambda$-quasi-twisted code is an
$S_\lambda$-invariant subspace of $\mathbb F_q^{2n}$; for $\lambda=1$ this
is the form used in \cite{Gal}. The reordering does not
change which entries $u_i$ and $v_i$ form a symplectic site, so the
symplectic product and symplectic weight have the same values in either
ordering. For two block vectors $(u\mid v)$ and $(a\mid b)$, the
symplectic product is
\[
\sum_{i=0}^{n-1}(u_ib_i-v_ia_i)=u\cdot b-v\cdot a,
\]
as given by the matrix $\Omega$ above.

Generator matrices of these codes are built from $\lambda$-circulant
matrices. For $a(x)=a_0+a_1x+\cdots+a_{n-1}x^{n-1}\in R_\lambda$, define
\[
\operatorname{circ}_\lambda(a)=
\begin{pmatrix}
a_0&a_1&\cdots&a_{n-1}\\
\lambda a_{n-1}&a_0&\cdots&a_{n-2}\\
\lambda a_{n-2}&\lambda a_{n-1}&\cdots&a_{n-3}\\
\vdots&\vdots&\ddots&\vdots\\
\lambda a_1&\lambda a_2&\cdots&a_0
\end{pmatrix}.
\]
Each row is the shift of the row above it. These matrices are called
twistulant in \cite{Nuh1}; for $\lambda=1$ they are the usual circulant
matrices. We use the following elementary fact throughout.

\begin{lemma}\label{circmult}
For $a,b\in R_\lambda$, the map
$\operatorname{circ}_\lambda:R_\lambda\to M_n(\mathbb F_q)$,
$a\mapsto\operatorname{circ}_\lambda(a)$, is an injective ring
homomorphism. In particular,
$\operatorname{circ}_\lambda(a)\operatorname{circ}_\lambda(b)=\operatorname{circ}_\lambda(ab)$.
The rows of $\operatorname{circ}_\lambda(a)$ represent
$a,xa,\ldots,x^{n-1}a$ in $R_\lambda$.
\end{lemma}

\begin{proof}
Since $\operatorname{circ}_\lambda(x)=T_\lambda^t$, right multiplication
of a coefficient row by $T_\lambda^t$ gives one constacyclic shift. We
have $(T_\lambda^t)^n=\lambda I_n$ and
$\operatorname{circ}_\lambda(a)=a(T_\lambda^t)$. Thus the map respects
multiplication modulo $x^n-\lambda$. The first row is the coefficient
vector of $a$, so the map is injective. The statement about the rows
follows from the definition of a $\lambda$-circulant matrix.
\end{proof}

The rows of
$\bigl(\operatorname{circ}_\lambda(a)\mid\operatorname{circ}_\lambda(b)\bigr)$
are $S_\lambda^{\,j}(a\mid b)$, $0\le j<n$, so they span the index-$2$
$\lambda$-quasi-twisted code generated by $(a,b)$.

\section{One-generator quasi-twisted codes}\label{sec:one-generator}

Let $R_\lambda=\mathbb F_q[x]/\langle x^n-\lambda\rangle$.
An index-$2$ $\lambda$-quasi-twisted code of length $2n$ is an
$R_\lambda$-submodule of $R_\lambda^2$.
We write a
one-generator $\lambda$-quasi-twisted code in the form $C=\langle(f_1(x)g_1(x),f_2(x)g_2(x))\rangle,$
where $g_i(x)$ is monic and
\[
g_i(x)\mid x^n-\lambda,\qquad
\gcd\left(f_i(x),\frac{x^n-\lambda}{g_i(x)}\right)=1,
\qquad i=1,2.
\]
Thus $C=\{(p(x)f_1(x)g_1(x),p(x)f_2(x)g_2(x)):
p(x)\in R_\lambda\}.$

\vskip 3pt
To obtain this form from any generating pair $(c_1(x),c_2(x))$, we take
\[
g_i(x)=\gcd(c_i(x),x^n-\lambda),\qquad
f_i(x)=\frac{c_i(x)}{g_i(x)},
\]
with monic greatest common divisors. Cancelling $g_i(x)$ in this
greatest-common-divisor identity gives
$\gcd(f_i(x),(x^n-\lambda)/g_i(x))=1$.

\begin{definition}

Let $C=\langle(f_1(x)g_1(x),f_2(x)g_2(x))\rangle$ have the form above.
The \emph{annihilator polynomial} of $C$ is the unique monic polynomial $h_C(x)$ of smallest degree satisfying
\[
h_C(x)f_i(x)g_i(x)\equiv0\pmod{x^n-\lambda},\qquad i=1,2.
\]

\end{definition}

The next proposition gives two forms of the same dimension formula.

\begin{proposition}\label{dimab}
Let $C=\langle(a(x),b(x))\rangle$ with $a(x)=f_1(x)g_1(x)$ and
$b(x)=f_2(x)g_2(x)$ in the form above, and put
$h_i(x)=(x^n-\lambda)/g_i(x)$ for $i=1,2$. Let
$\gamma(x)=\gcd\bigl(a(x),b(x),x^n-\lambda\bigr)$ be monic. Then
\begin{enumerate}[label=\textup{(\roman*)}]
\item $h_C(x)=\lcm\bigl(h_1(x),h_2(x)\bigr)=\dfrac{x^n-\lambda}{\gamma(x)}$;
\item $\dim_{\mathbb F_q}C=\deg h_C(x)=n-\deg\gamma(x)$.
\end{enumerate}
\end{proposition}

\begin{proof}
Define the surjective $\mathbb F_q$-linear map
\[
\Phi:R_\lambda\longrightarrow C,\qquad
\Phi(p(x))=\bigl(p(x)a(x),\,p(x)b(x)\bigr).
\]
Since $x^n-\lambda=g_i(x)h_i(x)$ and $\gcd\bigl(f_i(x),h_i(x)\bigr)=1$,
we have
\[
p(x)f_i(x)g_i(x)\equiv0\pmod{x^n-\lambda}
\quad\Longleftrightarrow\quad h_i(x)\mid p(x)f_i(x)
\quad\Longleftrightarrow\quad h_i(x)\mid p(x).
\]
Thus $\Phi(p(x))=0$ exactly when both $h_1(x)$ and $h_2(x)$ divide $p(x)$,
that is, when $\lcm\bigl(h_1(x),h_2(x)\bigr)$ divides $p(x)$. Hence
\[
h_C(x)=\lcm\bigl(h_1(x),h_2(x)\bigr)\qquad\text{and}\qquad
\ker\Phi=\langle h_C(x)\rangle.
\]

Since $\gcd\bigl(f_i(x),h_i(x)\bigr)=1$, we have
$\gcd\bigl(a(x),x^n-\lambda\bigr)=g_1(x)$ and
$\gcd\bigl(b(x),x^n-\lambda\bigr)=g_2(x)$, so
$\gamma(x)=\gcd\bigl(g_1(x),g_2(x)\bigr)$. Comparing the exponent of each
irreducible factor of $x^n-\lambda$, we get
\[
h_C(x)
=\lcm\left(\frac{x^n-\lambda}{g_1(x)},
           \frac{x^n-\lambda}{g_2(x)}\right)
=\frac{x^n-\lambda}{\gcd\bigl(g_1(x),g_2(x)\bigr)}
=\frac{x^n-\lambda}{\gamma(x)}.
\]
This proves (i).

Since $\Phi$ is surjective and $h_C(x)\mid x^n-\lambda$, we obtain
\[
C\cong R_\lambda/\langle h_C(x)\rangle
\cong\mathbb F_q[x]/\langle h_C(x)\rangle .
\]
Therefore,
$\dim_{\mathbb F_q}C=\deg h_C(x)=n-\deg\gamma(x)$. This proves (ii).
\end{proof}

\begin{corollary}\label{dim-common}
Let $g(x)$ be a monic divisor of $x^n-\lambda$, put
$h(x)=(x^n-\lambda)/g(x)$, and let
$C=\langle(r_1(x)g(x),\,r_2(x)g(x))\rangle$. If
$\gcd\bigl(r_1(x),r_2(x),h(x)\bigr)=1$, then $\dim_{\mathbb F_q}C=n-\deg g(x)$.
\end{corollary}

\begin{remark}
Let $C$ be a one-generator $\lambda$-quasi-twisted code. If $C^{\perp_s}\subseteq C$,
then $2n-\dim C\le\dim C$, while Proposition~\ref{dimab} gives
$\dim C\le n$. Hence $\dim C=n$ and $C=C^{\perp_s}$, so $C$ is symplectic
self-dual. Code O1 in Table~\ref{tab:qc-examples} is such an example with
$\lambda=1$.
\end{remark}

The next example illustrates the generator form and
the two dimension formulas for $\lambda=-1$.

\begin{example}

Take $R_{-1}=\mathbb F_3[x]/\langle x^{10}+1\rangle$. We have
\[
x^{10}+1=(x^2+1)(x^4+2x^3+x+1)(x^4+x^3+2x+1).
\]
Take
\[
\begin{aligned}
g_1(x)&=(x^2+1)(x^4+2x^3+x+1)
       =x^6+2x^5+x^4+x^2+x+1,\\
g_2(x)&=x^4+2x^3+x+1,
\end{aligned}
\]
and take $f_1(x)=f_2(x)=1$. Then
$C=\langle(g_1,g_2)\rangle$ has the required form, with
\[
\begin{aligned}
h_1(x)&=\frac{x^{10}+1}{g_1(x)}=x^4+x^3+2x+1,\\
h_2(x)&=\frac{x^{10}+1}{g_2(x)}
       =x^6+x^5+x^4+x^2+2x+1.
\end{aligned}
\]
Since $h_1\mid h_2$, Proposition~\ref{dimab} gives $h_C(x)=h_2(x)$
and $\dim C=6$. Also, $\gcd(g_1(x),g_2(x),x^{10}+1)=g_2(x)$, of degree $4$.
Proposition~\ref{dimab} therefore gives the same value,
$\dim C=10-4=6$.
The block generator matrix
$(\operatorname{circ}_{-1}(g_1)\mid\operatorname{circ}_{-1}(g_2))$
has rank $6$ as well.

\end{example}

\subsection{Circulant matrices and the symplectic dual}

Let $C\subseteq\mathbb F_q^{2n}$ be a linear code, and let $G$ be
a matrix whose rows span $C$. The entries of $G\Omega G^t$ are the
symplectic inner products of the rows of $G$. Since these rows span $C$,
we have
\[
C\subseteq C^{\perp_s}
\quad\Longleftrightarrow\quad
G\Omega G^t=0.
\]
We now apply this condition to one-generator quasi-twisted codes.

\begin{lemma}\label{t1}
Let $C=\langle(a(x),b(x))\rangle$ be a one-generator $\lambda$-quasi-twisted code, and put
$A=\operatorname{circ}_\lambda(a(x))$ and
$B=\operatorname{circ}_\lambda(b(x))$. Then $C$ is symplectic
self-orthogonal if and only if $AB^t=BA^t.$
\end{lemma}

\begin{proof}
We note that the matrix $G=(A\mid B)$ spans $C$.
Then
\[
\begin{aligned}
G\Omega G^t
&=\begin{pmatrix}A&B\end{pmatrix}
  \begin{pmatrix}O_n&I_n\\-I_n&O_n\end{pmatrix}
  \begin{pmatrix}A^t\\B^t\end{pmatrix}=\begin{pmatrix}-B&A\end{pmatrix}
  \begin{pmatrix}A^t\\B^t\end{pmatrix}=-BA^t+AB^t.
\end{aligned}
\]
Therefore, $G\Omega G^t=0$ if and only if $AB^t=BA^t$.
By the matrix criterion above, this is equivalent to
$C\subseteq C^{\perp_s}$.
\end{proof}

Lemma~\ref{t1} holds for every nonzero $\lambda$. To turn it into a
polynomial test, we describe the transpose of a $\lambda$-circulant matrix
in terms of polynomials. For $b(x)\in R_\lambda$, we define a polynomial
$b^{\dagger_\lambda}(x)$ whose $\lambda^{-1}$-circulant matrix is the
transpose of $\operatorname{circ}_\lambda(b(x))$. The subscript records the
twist $\lambda$ of $b(x)$, while $b^{\dagger_\lambda}(x)$ lies in
$R_{\lambda^{-1}}$.

\begin{proposition}\label{translambda}
Let $B=\operatorname{circ}_\lambda(b(x))$, where
$b(x)=\sum_{i=0}^{n-1}b_ix^i$. Define
\[
b^{\dagger_\lambda}(x)
=b_0+\lambda b_{n-1}x+\lambda b_{n-2}x^2+\cdots+\lambda b_1x^{n-1}.
\]
Then $B^t=\operatorname{circ}_{\lambda^{-1}}\bigl(b^{\dagger_\lambda}(x)\bigr)$.
\end{proposition}

\begin{proof}
Recall that
\[
B=\operatorname{circ}_\lambda(b)
=\begin{pmatrix}
b_0&b_1&\cdots&b_{n-1}\\
\lambda b_{n-1}&b_0&\cdots&b_{n-2}\\
\lambda b_{n-2}&\lambda b_{n-1}&\cdots&b_{n-3}\\
\vdots&\vdots&\ddots&\vdots\\
\lambda b_1&\lambda b_2&\cdots&b_0
\end{pmatrix},
~\text{so}\quad
B^t=\begin{pmatrix}
b_0&\lambda b_{n-1}&\cdots&\lambda b_1\\
b_1&b_0&\cdots&\lambda b_2\\
b_2&b_1&\cdots&\lambda b_3\\
\vdots&\vdots&\ddots&\vdots\\
b_{n-1}&b_{n-2}&\cdots&b_0
\end{pmatrix}.
\]

The matrix $B^t$ is $\lambda^{-1}$-circulant. Indeed, its first row is
$r_0=(b_0,\lambda b_{n-1},\ldots,\lambda b_1)$, and applying the
$\lambda^{-1}$-constacyclic shift gives the second row:
\[
T_{\lambda^{-1}}(r_0)
=(\lambda^{-1}(\lambda b_1),b_0,\lambda b_{n-1},\ldots,\lambda b_2)
=(b_1,b_0,\lambda b_{n-1},\ldots,\lambda b_2).
\]
The same argument applies to the remaining rows. Therefore $B^t$ is
$\lambda^{-1}$-circulant.

Its first row corresponds to the polynomial $b^{\dagger_\lambda}(x)
=b_0+\lambda b_{n-1}x+\cdots+\lambda b_1x^{n-1}.$
Hence $B^t=\operatorname{circ}_{\lambda^{-1}}\bigl(b^{\dagger_\lambda}\bigr)$.
\end{proof}

For $\lambda=1$, we have $b^{\dagger_1}(x)\equiv b(x^{-1})\pmod{x^n-1}$,
the familiar conjugate used for cyclic codes.

The transpose changes the twist to $\lambda^{-1}$. Abdukhalikov and
Verma~\cite{AbVerma2026} give explicit generators for the
symplectic dual of an index-$2$ quasi-twisted code. We give a short proof that $C^{\perp_s}$ is $\lambda^{-1}$-quasi-twisted, without
computing its generators.

\begin{proposition}\label{qt-dual-shift}
Let $C\subseteq\mathbb F_q^{2n}$ be a $\lambda$-quasi-twisted code of index $2$. Then
$C^{\perp_s}$ is a $\lambda^{-1}$-quasi-twisted code of index $2$.
\end{proposition}

\begin{proof}
Since $S_\lambda(C)=C$ and $S_\lambda$ is invertible, we have
$S_\lambda^{-1}(C)=C$. Take $U=(u\mid v)\in C^{\perp_s}$ and
$V=(a\mid b)\in C$, where $u,v,a,b\in\mathbb F_q^n$. Since
$T_{\lambda^{-1}}=T_\lambda^{-t}$ and
$\langle Mx,y\rangle_e=\langle x,M^ty\rangle_e$ for every matrix $M$,
we get
\[
\begin{aligned}
\langle S_{\lambda^{-1}}U,V\rangle_s
=\langle (T_{\lambda^{-1}}u\mid T_{\lambda^{-1}}v),(a\mid b)\rangle_s
&=\langle T_{\lambda^{-1}}u,b\rangle_e
  -\langle T_{\lambda^{-1}}v,a\rangle_e\\
&=\langle u,T_\lambda^{-1}b\rangle_e
  -\langle v,T_\lambda^{-1}a\rangle_e\\
&=\langle (u\mid v),(T_\lambda^{-1}a\mid T_\lambda^{-1}b)\rangle_s\\
&=\langle U,S_\lambda^{-1}V\rangle_s=0.
\end{aligned}
\]
The last expression is zero because $S_\lambda^{-1}V\in C$ and
$U\in C^{\perp_s}$. Since $V\in C$ was arbitrary,
$S_{\lambda^{-1}}U\in C^{\perp_s}$. Hence
$S_{\lambda^{-1}}(C^{\perp_s})\subseteq C^{\perp_s}$. Finally,
$S_{\lambda^{-1}}$ is invertible, so
$\dim S_{\lambda^{-1}}(C^{\perp_s})=\dim C^{\perp_s}$. Therefore the
inclusion is an equality, as required.
\end{proof}

\subsection{\texorpdfstring{The case $\lambda^2\ne1$}{The case lambda squared not equal to 1}}

When $\lambda^2\ne1$, the simultaneous constacyclic shift does not preserve the symplectic product. We first calculate the difference.

\begin{lemma}\label{sympdefect}
For vectors $U=(u\mid v)$ and $V=(u'\mid v')$ in $\mathbb F_q^{2n}$, the
simultaneous shift $S_\lambda$ satisfies
\[
\langle S_\lambda U,S_\lambda V\rangle_s-\langle U,V\rangle_s
=(\lambda^2-1)(u_{n-1}v'_{n-1}-v_{n-1}u'_{n-1}).
\]
\end{lemma}
\begin{proof}
By definition,
$S_\lambda U=(T_\lambda u\mid T_\lambda v)$ and 
$S_\lambda V=(T_\lambda u'\mid T_\lambda v').$
 Then
\[
\begin{aligned}
\langle S_\lambda U,S_\lambda V\rangle_s
&=\langle T_\lambda u,T_\lambda v'\rangle_e
  -\langle T_\lambda v,T_\lambda u'\rangle_e=\lambda^2(u_{n-1}v'_{n-1}-v_{n-1}u'_{n-1})
  +\sum_{i=0}^{n-2}(u_iv'_i-v_iu'_i).
\end{aligned}
\]
On the other hand, $\langle U,V\rangle_s
=(u_{n-1}v'_{n-1}-v_{n-1}u'_{n-1})
 +\sum_{i=0}^{n-2}(u_iv'_i-v_iu'_i).$
Subtracting the two expressions, we obtain $\langle S_\lambda U,S_\lambda V\rangle_s-\langle U,V\rangle_s
=(\lambda^2-1)(u_{n-1}v'_{n-1}-v_{n-1}u'_{n-1}),$
as required.
\end{proof}

\begin{theorem}\label{nogo}
Let $C$ be a nonzero $\mathbb F_q$-linear index-$2$ $\lambda$-quasi-twisted code, where
$\lambda^2\ne1$. Then $C\subseteq C^{\perp_s}$ if and only if there exist a
one-dimensional subspace $L=\mathbb F_q(\alpha,\beta)$ of $\mathbb F_q^2$ and
a nonzero $\lambda$-constacyclic code $D\subseteq\mathbb F_q^n$ such that
$C=\{(\alpha u\mid\beta u):u\in D\}$.
\end{theorem}

\begin{proof}
Suppose that $C\subseteq C^{\perp_s}$.
Take $U=(u\mid v)$ and $V=(u'\mid v')$ in $C$. Since $C$ is invariant
under $S_\lambda$, we have $S_\lambda^jU,\ S_\lambda^jV\in C$ for every
$j\ge0$. As $C\subseteq C^{\perp_s}$, both
$\langle S_\lambda^jU,S_\lambda^jV\rangle_s$ and
$\langle S_\lambda^{j+1}U,S_\lambda^{j+1}V\rangle_s$ are zero. Applying
Lemma~\ref{sympdefect} to
$S_\lambda^jU=(T_\lambda^ju\mid T_\lambda^jv)$ and
$S_\lambda^jV=(T_\lambda^ju'\mid T_\lambda^jv')$, we get
\[
0=(\lambda^2-1)\Bigl((T_\lambda^ju)_{n-1}(T_\lambda^jv')_{n-1}
-(T_\lambda^jv)_{n-1}(T_\lambda^ju')_{n-1}\Bigr),
\qquad 0\le j<n.
\]
Since $\lambda^2\ne1$, we may cancel the nonzero factor $\lambda^2-1$.
For $0\le j<n$,
\[
T_\lambda^j(a_0,\ldots,a_{n-1})
=(\lambda a_{n-j},\ldots,\lambda a_{n-1},\,a_0,\ldots,a_{n-1-j}),
\]
so the last coordinate of $T_\lambda^ja$ is $a_{n-1-j}$, with no factor
$\lambda$. Hence $u_{n-1-j}v'_{n-1-j}-v_{n-1-j}u'_{n-1-j}=0$ for
$0\le j<n$; that is, with $i=n-1-j$,
\begin{equation}\label{eq:det}
\det\begin{pmatrix}u_i&v_i\\u'_i&v'_i\end{pmatrix}=u_iv'_i-v_iu'_i=0,
\qquad 0\le i<n. \tag{$\ast$}
\end{equation}

For each $i$, let
$L_i=\{(u_i,v_i):(u\mid v)\in C\}\subseteq\mathbb F_q^2$. Since $C$ is
linear, $L_i$ is a subspace, and by~\eqref{eq:det} any two vectors of
$L_i$ are linearly dependent, so $\dim L_i\le1$.

Let $c=(u\mid v)\in C$, and write $p_i=(u_i,v_i)$ for its pair at
position $i$, so $p_i\in L_i$. The shifted word $S_\lambda c\in C$ has the
pairs $\lambda p_{n-1},p_0,\ldots,p_{n-2}$. Its pair at position $i+1$ is
$p_i$, so $L_i\subseteq L_{i+1}$ for $i<n-1$; its pair at position $0$ is
$\lambda p_{n-1}$, and since $\lambda\ne0$, this gives
$L_{n-1}\subseteq L_0$. Hence
\[
L_0\subseteq L_1\subseteq\cdots\subseteq L_{n-1}\subseteq L_0 .
\]
The chain starts and ends with $L_0$, so every inclusion is an equality,
and all $L_i$ equal one subspace $L$. Since $C\ne0$, we have $L\ne0$, and
since $\dim L\le1$, it is a line $L=\{\mu(\alpha,\beta):\mu\in\mathbb F_q\}$
with $(\alpha,\beta)\ne(0,0)$.

Let $c=(u\mid v)\in C$. Each pair $(u_i,v_i)$ lies in $L$, so
$(u_i,v_i)=z_i(\alpha,\beta)$ for some $z_i\in\mathbb F_q$. Putting
$z=(z_0,\ldots,z_{n-1})$ gives $c=(\alpha z\mid\beta z)$. The vector $z$ is
unique: if $(\alpha z\mid\beta z)=(\alpha z'\mid\beta z')$, then
$\alpha(z_i-z_i')=\beta(z_i-z_i')=0$, and $(\alpha,\beta)\ne(0,0)$ forces
$z=z'$. Let $D=\{z\in\mathbb F_q^n:(\alpha z\mid\beta z)\in C\}$. Then $D$
is a nonzero linear code and $C=\{(\alpha z\mid\beta z):z\in D\}$. For
$z\in D$,
\[
S_\lambda(\alpha z\mid\beta z)=(\alpha T_\lambda z\mid\beta T_\lambda z)\in C,
\]
so $T_\lambda z\in D$ by the definition of $D$. Hence $D$ is
$\lambda$-constacyclic.

Conversely, let $D$ be a nonzero $\lambda$-constacyclic code and
$C=\{(\alpha z\mid\beta z):z\in D\}$. For $z\in D$ we have
$S_\lambda(\alpha z\mid\beta z)=(\alpha T_\lambda z\mid\beta T_\lambda z)\in C$,
because $T_\lambda z\in D$; so $C$ is a nonzero $\lambda$-quasi-twisted code. For
$z,w\in\mathbb F_q^n$,
\[
\langle(\alpha z\mid\beta z),(\alpha w\mid\beta w)\rangle_s
=\alpha\beta\,(z\cdot w-z\cdot w)=0,
\]
so $C\subseteq C^{\perp_s}$.
\end{proof}

\begin{theorem}\label{nogo-distance}
Let $\lambda^2\ne1$, and let $C=\{(\alpha z\mid\beta z):z\in D\}$ be a
nonzero symplectic self-orthogonal $\lambda$-quasi-twisted code as in
Theorem~\ref{nogo}. Then the quantum code obtained from $C$ by
Theorem~\ref{Q} has minimum distance $1$, including the case of a quantum
code with dimension zero.
\end{theorem}

\begin{proof}
Consider the map $\psi:\mathbb F_q^n\to\mathbb F_q^{2n}$,
$\psi(z)=(\alpha z\mid\beta z)$. It is $\mathbb F_q$-linear and injective:
if $\psi(z)=0$, then $\alpha z=\beta z=0$, and since
$(\alpha,\beta)\ne(0,0)$, we get $z=0$. By definition $C=\psi(D)$, so
$\dim C=\dim D$.

Put $E=\psi(\mathbb F_q^n)$. Since $\psi$ is injective, $\dim E=n$. For
$z,w\in\mathbb F_q^n$, the calculation at the end of the proof of
Theorem~\ref{nogo} gives $\langle\psi(z),\psi(w)\rangle_s=0$, so
$E\subseteq E^{\perp_s}$. Since $\dim E^{\perp_s}=2n-\dim E=n=\dim E$, we
have $E=E^{\perp_s}$. Moreover, $C=\psi(D)\subseteq E$, so taking
symplectic duals gives $E=E^{\perp_s}\subseteq C^{\perp_s}$.

Let $e_i$ be the $i$th standard basis vector of $\mathbb F_q^n$. Then
$\psi(e_i)=(\alpha e_i\mid\beta e_i)\in E$, and its only nonzero pair is
$(\alpha,\beta)$ at position $i$, so $w_s(\psi(e_i))=1$.

\par\noindent
\emph{Case 1: $D\ne\mathbb F_q^n$.} Then $\dim C=\dim D<n$, and some $e_i$
does not lie in $D$. Since $E\subseteq C^{\perp_s}$, we have
$\psi(e_i)\in C^{\perp_s}$. If $\psi(e_i)\in C$, then
$\psi(e_i)=\psi(z)$ for some $z\in D$, and injectivity gives
$e_i=z\in D$, a contradiction. Hence $\psi(e_i)\in C^{\perp_s}\setminus C$
has symplectic weight $1$, and Theorem~\ref{Q} gives distance $1$.

\par\noindent
\emph{Case 2: $D=\mathbb F_q^n$.} Then $C=E=E^{\perp_s}=C^{\perp_s}$ and
$\dim C=n$. The vector $\psi(e_0)\in C$ is nonzero and has symplectic
weight $1$, so Theorem~\ref{Q} gives distance $d_s(C)=1$.
\end{proof}

\begin{corollary}
Let $C$ be an $\mathbb F_q$-linear index-$2$ $\lambda$-quasi-twisted code with
$\lambda^2\ne1$ and $C^{\perp_s}\subseteq C$. Then the quantum code of
Corollary~\ref{dc} has minimum distance $1$.
\end{corollary}

\begin{proof}
If $C^{\perp_s}=0$, then $C=\mathbb F_q^{2n}$, and every vector of
symplectic weight $1$ lies in $C\setminus C^{\perp_s}$, so the distance is
$1$. Otherwise, Proposition~\ref{qt-dual-shift} shows that $C^{\perp_s}$
is a nonzero symplectic self-orthogonal $\lambda^{-1}$-quasi-twisted code. Since
$(\lambda^{-1})^2\ne1$, Theorems~\ref{nogo} and~\ref{nogo-distance} apply
to $C^{\perp_s}$.
\end{proof}

Thus codes with distance larger than $1$ from the same-twist symplectic
construction require $\lambda^2=1$, that is, $\lambda=\pm1$ (only
$\lambda=1$ in characteristic $2$). The restriction concerns the same
scalar on both blocks. Hermitian and reciprocal mixed-twist constructions
use different boundary actions~\cite{LvQT2020,Ezerman2025}.
\subsection{\texorpdfstring{The case $\lambda^2=1$}{The case lambda squared equal to 1}}

Assume now that \(\lambda^2=1\). Then \(\lambda^{-1}=\lambda\), so
\(b^{\dagger_\lambda}(x)\in R_\lambda\) and $\operatorname{circ}_\lambda(b)^t
=\operatorname{circ}_\lambda(b^{\dagger_\lambda}).$
By Lemma~\ref{circmult}, $\operatorname{circ}_\lambda(a)\operatorname{circ}_\lambda(b)^t=\operatorname{circ}_\lambda(a\,b^{\dagger_\lambda})$.

\begin{theorem}\label{prf}
Let \(\lambda^2=1\), and let
\(C=\langle(a(x),b(x))\rangle\). Then
\(C\subseteq C^{\perp_s}\) if and only if
\[
a(x)b^{\dagger_\lambda}(x)
-b(x)a^{\dagger_\lambda}(x)
\equiv0\pmod{x^n-\lambda}.
\]
\end{theorem}

\begin{proof}
By Lemma~\ref{t1}, self-orthogonality is equivalent to \(AB^t=BA^t\). By Lemma~\ref{circmult} and Proposition~\ref{translambda},
$AB^t-BA^t=\operatorname{circ}_\lambda\bigl(a\,b^{\dagger_\lambda}
-b\,a^{\dagger_\lambda}\bigr)$. Since $\operatorname{circ}_\lambda$ is
injective, this matrix is zero if and only if
$a\,b^{\dagger_\lambda}-b\,a^{\dagger_\lambda}=0$ in $R_\lambda$.
\end{proof}

We now give a $\lambda=-1$ example for the polynomial criterion in Theorem~\ref{prf}.

\begin{example}
Take $R_{-1}=\mathbb F_3[x]/\langle x^8+1\rangle$. We factor $x^8+1=(x^4+x^2+2)(x^4+2x^2+2)$ in $\mathbb F_3[x]$.
Let $g(x)=x^4+x^2+2$, $r_1(x)=1$, and $r_2(x)=x^5+x^2$.
Set $C=\langle(r_1g,r_2g)\rangle$. Then
$h(x)=(x^8+1)/g(x)=x^4+2x^2+2$. Since $r_1=1$, we have
$\gcd(r_1,r_2,h)=1$. Hence $\dim C=8-4=4$.

Put $a(x)=r_1(x)g(x)=x^4+x^2+2$ and $b(x)=r_2(x)g(x)$. We obtain
\[
b(x)=x^7+x^6+2x^5+x^4+2x^2+2x.
\]
For $\lambda=-1$, Proposition~\ref{translambda} gives
$a^{\dagger_{-1}}(x)=2x^6+2x^4+2$ and
\[
b^{\dagger_{-1}}(x)=x^7+x^6+2x^4+x^3+2x^2+2x.
\]
In $R_{-1}$, direct multiplication gives
$a(x)b^{\dagger_{-1}}(x)-b(x)a^{\dagger_{-1}}(x)\equiv0\pmod{x^8+1}$.
Therefore Theorem~\ref{prf} shows that $C$ is symplectic self-orthogonal.

The block generator matrix has rank $4$, and a direct computation gives the classical parameters $[16,4,9]_3$. We also obtain
$\min\{w_s(v):v\in C^{\perp_s}\setminus C\}=3$.
Thus Theorem~\ref{Q} gives an $[[8,4,3]]_3$ quantum code. It is a quantum MDS code.
\end{example}

\section{Quantum codes from quasi-twisted and constacyclic codes}\label{sec:constructions}

As shown in Section~\ref{sec:one-generator}, a scalar twist with
$\lambda^2\ne1$ gives only quantum codes of distance $1$. We therefore use
$\lambda=1$ and $\lambda=-1$ for the scalar-twist constructions. We first give the one-generator quasi-cyclic examples O1--O6. We then construct T1--T4 from quasi-negacyclic codes, N1--N44 from quasi-cyclic codes, and H1--H13 from Hermitian constacyclic codes. The appendices give the generator coefficients of every labeled code.

\subsection{Examples of one-generator QC codes}

We give six one-generator QC codes $C=\langle(r_1g,\,r_2g)\rangle$ in
$R_1=\mathbb F_q[x]/\langle x^n-1\rangle$, where $g\mid x^n-1$ and
$\gcd\bigl(r_1,r_2,(x^n-1)/g\bigr)=1$. Then $\dim C=n-\deg g$. Each code
satisfies $AB^t=BA^t$, where $A=\operatorname{circ}_1(r_1g)$ and
$B=\operatorname{circ}_1(r_2g)$, so $C$ is symplectic self-orthogonal and
gives a quantum code by Theorem~\ref{Q}. Table~\ref{tab:qc-examples} gives
the parameters, and Table~\ref{tab:qc-example-generators} gives the
coefficients.

Codes O1--O4 first appeared in an earlier version of this
work~\cite{BagV1}, and their parameters are now recorded in Grassl's
tables~\cite{gra}. The parameters $[[30,3,10]]_3$ of O5 already appear
in~\cite{ADL2026}; O5 gives a one-generator construction of
the same parameters. Code O6, with parameters $[[35,13,8]]_3$, improves
the recorded lower bound $7$.

\begin{table}[H]
\centering\small\twocolsetup
\caption{One-generator QC examples, all with $\lambda=1$.}\label{tab:qc-examples}
\begin{tabular}{@{}lll@{\hspace{1.2em}}V@{\hspace{1.2em}}lll@{}}
\toprule
No. & Quantum code & Classical code $C$ & No. & Quantum code & Classical code $C$\\
\midrule
O1 & $[[11,0,6]]_5$  & $[22,11]_5$ & O4 & $[[16,6,5]]_3$  & $[32,10]_3$\\
O2 & $[[13,6,4]]_3$  & $[26,7]_3$  & O5 & $[[30,3,10]]_3$ & $[60,27]_3$\\
O3 & $[[23,12,5]]_3$ & $[46,11]_3$ & O6 & $[[35,13,8]]_3$ & $[70,22]_3$\\
\bottomrule
\end{tabular}
\end{table}

\subsection{New codes from scalar twists}

The codes T1--T4 and N1--N44 are one-generator quasi-twisted codes
$C=\langle(g,\,r_2g)\rangle\subseteq R_\lambda^2$, where $g$ is a monic
divisor of $x^n-\lambda$ and $r_2\in R_\lambda$. Since the first entry is
$g$ itself, Corollary~\ref{dim-common} gives $\dim C=n-\deg g$. Writing
$k=\deg g$, the code $C$ has parameters $[2n,n-k]_q$. For each code we
verified $AB^t=BA^t$, where $A=\operatorname{circ}_\lambda(g)$ and
$B=\operatorname{circ}_\lambda(r_2g)$. Hence $C$ is symplectic
self-orthogonal, and Theorem~\ref{Q} gives a quantum code with parameters
$[[n,k,d]]_q$. Each of these codes improves the lower bound recorded in
Grassl's tables~\cite{gra}. The columns
``Previous'' give the lower and upper bounds in~\cite{gra}.

\paragraph{Quantum codes from quasi-negacyclic codes.} For T1--T4 we take
$\lambda=-1$. Table~\ref{tab:negacyclic} gives the parameters, and
Table~\ref{tab:qn-generators} gives the generator polynomials $g$ and $r_2$.

\begin{table}[H]
\centering\small\twocolsetup
\setlength{\tabcolsep}{4pt}
\caption{New quantum codes from quasi-negacyclic codes.}\label{tab:negacyclic}
\begin{tabular}{@{}lllc@{\hspace{1.2em}}V@{\hspace{1.2em}}lllc@{}}
\toprule
No. & Quantum code & Classical code $C$ & Previous &
No. & Quantum code & Classical code $C$ & Previous\\
\midrule
T1 & $[[84,62,7]]_5$ & $[168,22]_5$ & 6--11 & T3 & $[[50,16,13]]_7$ & $[100,34]_7$ & 12--18\\
T2 & $[[60,36,9]]_7$ & $[120,24]_7$ & 8--13 & T4 & $[[52,8,14]]_3$  & $[104,44]_3$ & 13--20\\
\bottomrule
\end{tabular}
\end{table}

\paragraph{Quantum codes from quasi-cyclic codes.} For N1--N44 we take $\lambda=1$.
Table~\ref{tab:qc-direct} gives the parameters, and
Table~\ref{tab:qc-generators} gives the generator polynomials $g$ and $r_2$. A dagger
marks a code whose parameters also follow from another code in this paper
by dimension reduction; see Table~\ref{tab:dominated}. For
these parameters, the daggered code gives an explicit generator instead
of an existence result from propagation.

\begin{table}[H]
\centering\small\twocolsetup
\setlength{\tabcolsep}{5pt}
\caption{New quantum codes from quasi-cyclic codes.  A dagger marks
a row listed in Table~\ref{tab:dominated}.}\label{tab:qc-direct}
\begin{tabular}{@{}llc@{\hspace{1em}}V@{\hspace{1em}}llc@{}}
\toprule
No. & Quantum code & Previous~\cite{gra} & No. & Quantum code & Previous~\cite{gra}\\
\midrule
N1 & $[[52,31,7]]_3$ & 6--9 & N23 & $[[36,7,12]]_7$ & 11--15\\
N2 & $[[31,15,7]]_5$ & 6--9 & N24 & $[[37,10,11]]_7$ & 10--14\\
N3 & $[[31,12,8]]_5$ & 7--10 & N25 & $[[38,11,11]]_7$ & 10--14\\
N4 & $[[71,51,7]]_5$ & 6--10 & N26 & $[[38,19,8]]_7$ & 7--10\\
N5 & $[[71,46,8]]_5$ & 7--13 & N27 & $[[40,13,11]]_7$ & 10--14\\
N6 & $[[71,45,9]]_5$ & 8--13 & N28 & $[[43,24,8]]_7$ & 7--10\\
N7 & $[[78,51,9]]_5$ & 8--14 & N29 & $[[57,34,9]]_7$ & 8--12\\
N8 & $[[71,41,10]]_5$ & 9--15 & N30 & $[[48,19,11]]_7$ & 10--15\\
N9$^\dagger$ & $[[71,40,10]]_5$ & 9--16 & N31 & $[[55,25,11]]_7$ & 10--16\\
N10$^\dagger$ & $[[91,70,7]]_5$ & 6--10 & N32$^\dagger$ & $[[57,26,11]]_7$ & 10--16\\
N11 & $[[91,71,7]]_5$ & 6--10 & N33 & $[[57,27,11]]_7$ & 10--16\\
N12$^\dagger$ & $[[93,64,9]]_5$ & 8--14 & N34 & $[[57,25,12]]_7$ & 11--17\\
N13 & $[[93,65,9]]_5$ & 8--14 & N35 & $[[48,11,14]]_7$ & 13--19\\
N14 & $[[93,68,8]]_5$ & 7--12 & N36 & $[[55,20,13]]_7$ & 12--18\\
N15$^\dagger$ & $[[93,70,7]]_5$ & 6--11 & N37 & $[[48,6,15]]_5$ & 14--21\\
N16$^\dagger$ & $[[93,71,7]]_5$ & 6--11 & N38 & $[[48,9,15]]_7$ & 14--20\\
N17 & $[[93,72,7]]_5$ & 6--10 & N39$^\dagger$ & $[[50,9,15]]_7$ & 14--21\\
N18 & $[[52,27,8]]_3$ & 7--12 & N40$^\dagger$ & $[[52,10,15]]_5$ & 14--21\\
N19 & $[[31,24,4]]_5$ & 3--4 & N41 & $[[48,3,16]]_5$ & 15--23\\
N20 & $[[44,17,10]]_5$ & 9--14 & N42 & $[[48,6,16]]_7$ & 15--22\\
N21 & $[[29,8,9]]_7$ & 8--11 & N43 & $[[57,19,14]]_7$ & 13--20\\
N22 & $[[32,9,10]]_7$ & 9--12 & N44$^\dagger$ & $[[50,8,16]]_7$ & 15--22\\
\bottomrule
\end{tabular}
\end{table}

\subsection{Codes from a matrix twist}

This subsection differs from the scalar-twist setting of
Section~\ref{sec:one-generator} in two ways.
First, the shift acts on pairs rather than on single coordinates. For a
scalar twist we used the $\lambda$-constacyclic shift
$T_\lambda:\mathbb F_q^n\to\mathbb F_q^n$. Here we use a shift
$T_\Lambda:(\mathbb F_q^2)^n\to(\mathbb F_q^2)^n$, acting on $n$-tuples of
coordinate pairs $(u_i,v_i)^t$, whose twist is a matrix
$\Lambda\in M_2(\mathbb F_q)$ instead of a scalar
$\lambda\in\mathbb F_q^*$.

Second, the induced shift on $\mathbb F_q^{2n}$ no longer acts on the two
blocks separately. For a scalar twist,
$S_\lambda(u\mid v)=(T_\lambda u\mid T_\lambda v)$, so
$S_\lambda=T_\lambda\times T_\lambda$ shifts $u$ and $v$ independently. When $\Lambda$ is not diagonal, $S_\Lambda$ no longer acts on the
two blocks separately: at the boundary, at least one block receives
a contribution from the other block. Let
\[
\Lambda=\begin{pmatrix}\ell_{11}&\ell_{12}\\\ell_{21}&\ell_{22}\end{pmatrix}
\in M_2(\mathbb F_q).
\]
We call $\Lambda$ the \emph{twist matrix}. We write a vector
$U=(u\mid v)\in\mathbb F_q^{2n}$, with $u=(u_0,\ldots,u_{n-1})$ and
$v=(v_0,\ldots,v_{n-1})$, in two equivalent forms. In \emph{column form},
$U$ is the $2\times n$ array
\[
Z=\begin{pmatrix}u_0&u_1&\cdots&u_{n-1}\\v_0&v_1&\cdots&v_{n-1}\end{pmatrix}
=(z_0\ z_1\ \cdots\ z_{n-1}),
\qquad z_i=\begin{pmatrix}u_i\\v_i\end{pmatrix}.
\]
In \emph{row form}, $U$
is the block vector $(u\mid v)$, obtained by reading the two rows of $Z$
one after the other.

The matrix-twist shift $T_\Lambda:(\mathbb F_q^2)^n\to(\mathbb F_q^2)^n$
acts on the column form. It moves each column one place to the right and
sends the last column to the front, multiplied by $\Lambda$:
\[
T_\Lambda(Z)=(\Lambda z_{n-1}\ z_0\ \cdots\ z_{n-2})
=\begin{pmatrix}
\ell_{11}u_{n-1}+\ell_{12}v_{n-1}&u_0&\cdots&u_{n-2}\\
\ell_{21}u_{n-1}+\ell_{22}v_{n-1}&v_0&\cdots&v_{n-2}
\end{pmatrix}.
\]
Reading the two rows of $T_\Lambda(Z)$ gives the same shift in row form,
namely the map $S_\Lambda:\mathbb F_q^{2n}\to\mathbb F_q^{2n}$ given by
\[
\begin{aligned}
S_\Lambda(u\mid v)
=(&\ell_{11}u_{n-1}+\ell_{12}v_{n-1},u_0,\ldots,u_{n-2}\mid\ell_{21}u_{n-1}+\ell_{22}v_{n-1},v_0,\ldots,v_{n-2}).
\end{aligned}
\]
Thus $T_\Lambda$ and $S_\Lambda$ are the same shift, written in column
form and in row form, respectively. If
$\Lambda=\operatorname{diag}(\lambda_1,\lambda_2)$, then
$S_\Lambda=T_{\lambda_1}\times T_{\lambda_2}$, and for $\Lambda=\lambda I_2$
we recover the usual $\lambda$-quasi-twisted shift.

For column vectors $z=(a,b)^t$ and $w=(c,d)^t$ in $\mathbb F_q^2$, we write
\[
\langle z,w\rangle_s=z^tJw=ad-bc=\det(z\ \ w),
\qquad J=\begin{pmatrix}0&1\\-1&0\end{pmatrix},
\]
which is the symplectic inner product of
Section~\ref{sec:preliminaries} for $n=1$. For
$Z=(z_0,\ldots,z_{n-1})$ and $W=(w_0,\ldots,w_{n-1})$, we get
\[
\langle Z,W\rangle_s=\sum_{i=0}^{n-1}\langle z_i,w_i\rangle_s .
\]
This is the symplectic inner product of the corresponding block vectors.
We first determine when the matrix shift preserves this product.

\begin{lemma}\label{boundary-form}
Let $\Lambda\in M_2(\mathbb F_q)$. For $z,w\in\mathbb F_q^2$,
\[
\langle\Lambda z,\Lambda w\rangle_s=(\det\Lambda)\langle z,w\rangle_s.
\]
For $Z=(z_0,\ldots,z_{n-1})$ and $W=(w_0,\ldots,w_{n-1})$ in
$(\mathbb F_q^2)^n$,
\[
\langle T_\Lambda(Z),T_\Lambda(W)\rangle_s-\langle Z,W\rangle_s
=(\det\Lambda-1)\langle z_{n-1},w_{n-1}\rangle_s .
\]
Consequently, $T_\Lambda$ preserves the symplectic inner product if and
only if $\det\Lambda=1$.
\end{lemma}

\begin{proof}
Since $\langle z,w\rangle_s=\det(z\ \ w)$, we have
$\langle\Lambda z,\Lambda w\rangle_s=\det\bigl(\Lambda(z\ \ w)\bigr)
=(\det\Lambda)\langle z,w\rangle_s$.
\vskip 3pt \noindent
For $Z$ and $W$, the definition of $T_\Lambda$ gives
\begin{align*}
\langle T_\Lambda(Z),T_\Lambda(W)\rangle_s
&=\langle \Lambda z_{n-1},\Lambda w_{n-1}\rangle_s+\sum_{i=0}^{n-2}\langle z_i,w_i\rangle_s=(\det \Lambda)\langle z_{n-1},w_{n-1}\rangle_s+\sum_{i=0}^{n-2}\langle z_i,w_i\rangle_s.
\end{align*}
On the other hand, we have $\langle Z,W\rangle_s
=\langle z_{n-1},w_{n-1}\rangle_s+\sum_{i=0}^{n-2}\langle z_i,w_i\rangle_s.$

Subtracting the two expressions gives
\begin{align*}
&\langle T_\Lambda(Z),T_\Lambda(W)\rangle_s-\langle Z,W\rangle_s=(\det \Lambda-1)\langle z_{n-1},w_{n-1}\rangle_s.
\end{align*}
If $\det \Lambda=1$, the difference $(\det \Lambda-1)$ is zero, so $T_\Lambda$ preserves
the symplectic inner product. 
\vskip 5pt
Conversely, suppose that $T_\Lambda$ preserves the symplectic inner product for all
$Z,W$. Then
\[
  \langle T_\Lambda(Z),T_\Lambda(W)\rangle_s=\langle Z,W\rangle_s \implies (\det \Lambda - 1)\,\langle z_{n-1}, w_{n-1}\rangle_s = 0.
\]
Taking $z_{n-1} = (1,0)^t$ and $w_{n-1} = (0,1)^t$, we have
$\langle z_{n-1}, w_{n-1}\rangle_s = 1$. Thus $\det \Lambda - 1 = 0$, which gives
$\det \Lambda = 1$.
\end{proof}

\vskip 1pt
\noindent
For $\Lambda=\lambda I_2$, we have $\det \Lambda=\lambda^2$.
Thus Lemma~\ref{boundary-form} also gives the identity in
Lemma~\ref{sympdefect}.  From now on, we take the twist matrix to be
\[
\Lambda=M_t=\begin{pmatrix}0&-1\\1&t\end{pmatrix},
\qquad t\in\mathbb F_q,
\]
and write $T_{M_t}$ and $S_{M_t}$ for the corresponding shifts. Then $\det M_t=1$, and for $z_i=(u_i,v_i)^t$ we have
$M_tz_i=(-v_i,\,u_i+tv_i)^t$. Applying this to $z_{n-1}$ gives
\[
S_{M_t}(u\mid v)
=(-v_{n-1},u_0,\ldots,u_{n-2}\mid u_{n-1}+tv_{n-1},v_0,\ldots,v_{n-2}).
\]
Since $\det M_t=1$, Lemma~\ref{boundary-form} shows that this shift
preserves the symplectic inner product. Since $M_t$ is not diagonal,
$S_{M_t}$ does not act on the two blocks separately.
\vskip 2pt
We call a linear code $C_{M_t}\subseteq\mathbb F_q^{2n}$ an
\emph{$M_t$-twisted code} if $S_{M_t}(C_{M_t})=C_{M_t}$. Such a code need
not be symplectic self-orthogonal. We next describe these codes using
polynomials and give a condition for self-orthogonality.

\subsubsection{Polynomial description}
In Sections~\ref{sec:preliminaries} and~\ref{sec:one-generator}, we represented
$(u\mid v)$ by the pair $(u(x),v(x))\in R_\lambda^2$. This was natural for
the scalar shift $S_\lambda$: multiplication by $x$ shifts each block
separately, and the relation $x^n=\lambda$ handles the last coordinate of
each block. The matrix shift $S_{M_t}$ is different. By its definition
above, $u_{n-1}$ contributes to the first coordinate of the second block,
while $v_{n-1}$ contributes to the first coordinates of both blocks. We
therefore use a single polynomial in which multiplication by $x$ produces
this shift.

For $(u\mid v)\in\mathbb F_q^{2n}$, put
$u(x)=\sum_{i=0}^{n-1}u_ix^i$ and $v(x)=\sum_{i=0}^{n-1}v_ix^i$, and place
the coefficients of $v(x)$ immediately after those of $u(x)$:
\[
c(x)=u(x)+x^nv(x)
=u_0+\cdots+u_{n-1}x^{n-1}+v_0x^n+\cdots+v_{n-1}x^{2n-1}.
\]
Thus the pair $(u_i,v_i)$ consists of the coefficients of $x^i$ and
$x^{n+i}$, and multiplication by $x$ moves $u_{n-1}x^{n-1}$ to
$u_{n-1}x^n$, the first position of the second block.

We now find a polynomial $m_t(x)$ such that multiplication by $x$ modulo
$m_t(x)$ gives the shift $S_{M_t}$. By the definition of $S_{M_t}$, the
shifted vector has the polynomial
\[
c_{\mathrm{shift}}(x)
=-v_{n-1}+\sum_{i=0}^{n-2}u_ix^{i+1}+(u_{n-1}+tv_{n-1})x^n
+\sum_{i=0}^{n-2}v_ix^{n+i+1}.
\]
On the other hand,
\[
xc(x)=\sum_{i=0}^{n-2}u_ix^{i+1}+u_{n-1}x^n
+\sum_{i=0}^{n-2}v_ix^{n+i+1}+v_{n-1}x^{2n}.
\]
The two expressions differ only in the terms involving $v_{n-1}$:
\[
xc(x)-c_{\mathrm{shift}}(x)=v_{n-1}(x^{2n}-tx^n+1).
\]
We therefore put
\[
m_t(x)=x^{2n}-tx^n+1,\qquad R_{m_t}=\mathbb F_q[x]/\langle m_t(x)\rangle .
\]
Since $\deg m_t=2n$, every class in $R_{m_t}$ has a unique representative
of degree less than $2n$, so $(u\mid v)\mapsto c(x)$ identifies
$\mathbb F_q^{2n}$ with $R_{m_t}$. Since $x^{2n}=tx^n-1$ in $R_{m_t}$,
\[
xc(x)\equiv c_{\mathrm{shift}}(x)\pmod{m_t(x)}.
\]
Thus multiplication by $x$ in $R_{m_t}$ is exactly the matrix-twisted
shift $S_{M_t}$.

\subsubsection{\texorpdfstring{$M_t$}{Mt}-twisted code and self-orthogonality}
We now use this representation to describe $M_t$-twisted codes. If
$C_{M_t}$ is an $M_t$-twisted code and $c(x)\in C_{M_t}$, then
$xc(x)\in C_{M_t}$; repeated multiplication by $x$ and linearity give
$p(x)c(x)\in C_{M_t}$ for every $p(x)\in R_{m_t}$. Thus $C_{M_t}$ is an
ideal of $R_{m_t}$. Conversely, every ideal of $R_{m_t}$ is closed under
multiplication by $x$ and hence is an $M_t$-twisted code. Since
$\mathbb F_q[x]$ is a principal ideal domain, every ideal of $R_{m_t}$ is
generated by a unique monic divisor $g(x)$ of $m_t(x)$, as in the cyclic
case. We write
\[
C_{M_t}=\langle g(x)\rangle,\qquad m_t(x)=g(x)h(x),
\]
and call $g(x)$ the generator polynomial and $h(x)$ the check polynomial.

\begin{lemma}\label{matrix-ideal-basis}
Let $g$ be a monic divisor of $m_t$, put $h=m_t/g$, and let $r=\deg h$.
The coefficient vectors of $g,xg,\ldots,x^{r-1}g$ form a basis of
$C_{M_t}=\langle g\rangle$, with the empty family when $r=0$. In
particular, $\dim_{\mathbb F_q}C_{M_t}=r$.
\end{lemma}

\begin{proof}
If $r=0$, then $g=m_t$ and $C_{M_t}=\{0\}$, so the result holds.
Suppose that $r>0$. Every codeword has the form $p(x)g(x)$.
Dividing $p(x)$ by $h(x)$ gives
\[
p(x)=s(x)h(x)+a(x),\qquad \deg a<r.
\]
Since $h(x)g(x)=m_t(x)$, this codeword equals $a(x)g(x)$ in $R_{m_t}$.
Thus the vectors $g,xg,\ldots,x^{r-1}g$ span the code.

If $a(x)g(x)=0$ in $R_{m_t}$, then $g(x)h(x)$ divides $a(x)g(x)$.
Cancelling $g(x)$ gives $h(x)\mid a(x)$. Since $\deg a<r=\deg h$, we
have $a=0$. Hence the vectors are linearly independent.
\end{proof}

We use the symplectic dual defined in Section~\ref{sec:preliminaries}.
Through the identification $(u\mid v)\leftrightarrow c(x)=u(x)+x^nv(x)$,
we regard $C_{M_t}^{\perp_s}$ as a subset of $R_{m_t}$. We first show that
it is also $M_t$-twisted.

\begin{proposition}
Let $C_{M_t}\subseteq\mathbb F_q^{2n}$ be an $M_t$-twisted code. Then
$C_{M_t}^{\perp_s}$ is also an $M_t$-twisted code.
\end{proposition}

\begin{proof}
Since $C_{M_t}$ is $M_t$-twisted, $S_{M_t}(C_{M_t})=C_{M_t}$. Since
$S_{M_t}$ is invertible, also $S_{M_t}^{-1}(C_{M_t})=C_{M_t}$.
Let $w\in C_{M_t}^{\perp_s}$ and $c\in C_{M_t}$. By
Lemma~\ref{boundary-form}, $S_{M_t}$ preserves the symplectic inner
product, so
\[
\langle S_{M_t}w,c\rangle_s=\langle w,S_{M_t}^{-1}c\rangle_s=0,
\]
because $S_{M_t}^{-1}c\in C_{M_t}$ and $w\in C_{M_t}^{\perp_s}$. Since $c$
and $w$ were arbitrary, $S_{M_t}\bigl(C_{M_t}^{\perp_s}\bigr)\subseteq
C_{M_t}^{\perp_s}$. Finally, $S_{M_t}$ is invertible, so it preserves
dimensions, and the inclusion is an equality. Therefore
$C_{M_t}^{\perp_s}$ is $M_t$-twisted.
\end{proof}

For $f(0)\ne0$, define the monic reciprocal polynomial by
\[
f^*(x)=f(0)^{-1}x^{\deg f}f(x^{-1}).
\]
The constant terms of $g$ and $h$ are nonzero because
$g(0)h(0)=m_t(0)=1$. The reciprocal of a product is the product of the
reciprocals:
\[
(gh)^*(x)=\frac{x^{\deg g+\deg h}g(x^{-1})h(x^{-1})}{g(0)h(0)}
=g^*(x)h^*(x).
\]
Since $m_t(0)=1$, we also have
\[
m_t^*(x)=x^{2n}m_t(x^{-1})=1-tx^n+x^{2n}=m_t(x).
\]
Thus $m_t=g^*h^*$ whenever $m_t=gh$. We use these identities to find the
symplectic dual.

\begin{proposition}\label{matrix-ideal-dual}
Let $C_{M_t}=\langle g\rangle$ and $m_t=gh$ as above. Then
$C_{M_t}^{\perp_s}=\langle h^*\rangle$.
\end{proposition}

\begin{proof}
We identify each polynomial in $R_{m_t}$ with its coefficient vector in
$\mathbb F_q^{2n}$, and use the same letter for both; thus
$\langle a,w\rangle_s$ is the symplectic product of the coefficient
vectors of $a$ and $w$. By the polynomial description of the shift, the
coefficient vector of $xa$ is $S_{M_t}a$. Since $\det M_t=1$,
Lemma~\ref{boundary-form} shows that $S_{M_t}$ preserves the symplectic
inner product. Hence, for all $a,w\in R_{m_t}$,
\[
\langle xa,xw\rangle_s=\langle S_{M_t}a,S_{M_t}w\rangle_s=\langle a,w\rangle_s .
\]
Next, $x$ is invertible in $R_{m_t}$. Indeed,
\[
x\,(t x^{n-1}-x^{2n-1})=tx^n-x^{2n}\equiv1\pmod{m_t},
\]
so $x^{-1}=tx^{n-1}-x^{2n-1}$ in $R_{m_t}$. Replacing $w$ by $x^{-1}w$ in
the identity above, we obtain
\[
\langle xa,w\rangle_s=\langle xa,x(x^{-1}w)\rangle_s=\langle a,x^{-1}w\rangle_s .
\]
Applying this identity $j$ times gives
$\langle x^ja,w\rangle_s=\langle a,x^{-j}w\rangle_s$ for $j\ge0$. Let
$p(x)=\sum_j p_jx^j\in\mathbb F_q[x]$. By linearity,
\[
\langle p(x)a,w\rangle_s=\sum_j p_j\langle x^ja,w\rangle_s
=\sum_j p_j\langle a,x^{-j}w\rangle_s=\langle a,p(x^{-1})w\rangle_s .
\]

Every codeword of $C_{M_t}$ has the form $g(x)a$ for some
$a\in R_{m_t}$. Therefore $w\in C_{M_t}^{\perp_s}$ if and only if
$\langle g(x)a,w\rangle_s=0$ for all $a\in R_{m_t}$. Taking $p=g$ above,
this condition becomes $\langle a,g(x^{-1})w\rangle_s=0$ for all
$a\in R_{m_t}$. Since $a$ runs through all vectors of $\mathbb F_q^{2n}$
and the symplectic inner product is nondegenerate, this holds if and only
if $g(x^{-1})w=0$ in $R_{m_t}$.

By the definition of the reciprocal polynomial,
$g(x^{-1})=g(0)\,x^{-\deg g}\,g^*(x)$. Here $g(0)\ne0$ and $x$ is
invertible in $R_{m_t}$, so $g(0)\,x^{-\deg g}$ is a unit of $R_{m_t}$.
Hence $g(x^{-1})w=0$ if and only if $g^*(x)w=0$ in $R_{m_t}$. Finally,
since $m_t=g^*h^*$,
\[
\begin{aligned}
g^*(x)w=0\ \text{in }R_{m_t}
&\iff g^*(x)h^*(x)\mid g^*(x)w(x)\iff h^*(x)\mid w(x).
\end{aligned}
\]
Therefore $C_{M_t}^{\perp_s}=\langle h^*\rangle$.
\end{proof}

In particular, $C_{M_t}^{\perp_s}$ is again an $M_t$-twisted code. Since
$g$ and $h^*$ are monic divisors of $m_t$,
\[
C_{M_t}\subseteq C_{M_t}^{\perp_s}
\iff\langle g\rangle\subseteq\langle h^*\rangle
\iff h^*\mid g .
\]
Thus $h^*\mid g$ is the polynomial test for symplectic self-orthogonality.
Alternatively, one may check $G\Omega G^t=0$ for the generator matrix $G$
of Lemma~\ref{matrix-ideal-basis}.

Now suppose that $h^*\mid g$, so that $C_{M_t}\subseteq C_{M_t}^{\perp_s}$.
By Lemma~\ref{matrix-ideal-basis}, $\dim C_{M_t}=r$. Since
$\dim C_{M_t}^{\perp_s}=2n-r$, the inclusion gives $r\le 2n-r$, that is,
$r\le n$. If $r<n$, Theorem~\ref{Q} gives a quantum code with parameters
$[[n,n-r,d]]_q$, where
$d=\min\{w_s(c):c\in C_{M_t}^{\perp_s}\setminus C_{M_t}\}$. If $r=n$, then
$C_{M_t}=C_{M_t}^{\perp_s}$, and Theorem~\ref{Q} gives a quantum code
$[[n,0,d]]_q$ with $d=d_s(C_{M_t})$. 

\begin{example}
Let $q=3$, $t=0$, and $n=2$. Over $\mathbb F_3$,
\[
m_0(x)=x^4+1=(x^2+x+2)(x^2+2x+2).
\]
Take $g(x)=x^2+x+2$ and $h(x)=x^2+2x+2$. Since $h(0)=2$,
\[
h^*(x)=2^{-1}x^2h(x^{-1})=x^2+x+2=g(x).
\]
Thus $C_{M_t}=\langle g(x)\rangle=C_{M_t}^{\perp_s}$, and Lemma~\ref{matrix-ideal-basis}
gives $\dim C_{M_t}=\deg h(x)=2$. Every codeword is $a g(x)+b x g(x)$ with $a,b\in\mathbb F_3$. Since
$g(x)=2+x+x^2$ and $xg(x)=2x+x^2+x^3$ have degree less than $4$, its
coefficient vector is $(2a,\;a+2b,\;a+b,\;b)$. Hence
$u=(2a,\,a+2b)$ and $v=(a+b,\,b)$, so its two coordinate pairs are
$(u_0,v_0)=(2a,\,a+b)$ and $(u_1,v_1)=(a+2b,\,b)$. If either pair is zero, then
$a=b=0$. Thus every nonzero codeword has symplectic weight $2$,
and Theorem~\ref{Q} gives a $[[2,0,2]]_3$ code.
\end{example}

\subsubsection{Quadratic-extension description}

So far we have described $C_{M_t}$ and its symplectic dual over
$\mathbb F_q$. When $z^2-tz+1$ is irreducible over $\mathbb F_q$, the same
codes have a second description over $\mathbb F_{q^2}$: each coordinate
pair $(u_i,v_i)\in\mathbb F_q^2$ becomes a single element of
$\mathbb F_{q^2}$, and the matrix twist $M_t$ becomes multiplication by a
scalar $\lambda\in\mathbb F_{q^2}$. 

The description over
$\mathbb F_{q^2}$ relates these codes to the known norm-one Hermitian
constacyclic construction, and it lets us compute symplectic weights as
Hamming weights over $\mathbb F_{q^2}$. The matrix description over
$\mathbb F_q$ is the one we use for the codes themselves: the generators
$g$ and $h$, the dual generator $h^*$, and the test $h^*\mid g$ are
polynomials over $\mathbb F_q$, and they apply directly to the symplectic
construction of Theorem~\ref{Q}, without passing through the Hermitian
form.

Throughout this subsection, $\lambda$ denotes a root of $z^2-tz+1$ in
$\mathbb F_{q^2}$. It should not be confused with the scalar twist
$\lambda\in\mathbb F_q^*$ of
Sections~\ref{sec:preliminaries} and~\ref{sec:one-generator}.

\begin{lemma}\label{lambda-facts}
Suppose that $z^2-tz+1$ is irreducible over $\mathbb F_q$, and let
$\lambda\in\mathbb F_{q^2}$ be a root. Then the following hold.
\begin{enumerate}[label=\textup{(\roman*)}]
\item The elements $1,\lambda$ form a basis of $\mathbb F_{q^2}$ over
$\mathbb F_q$. Thus every element of $\mathbb F_{q^2}$ can be written
uniquely as $u+\lambda v$ with $u,v\in\mathbb F_q$.
\item The other root of $z^2-tz+1$ is $\lambda^q$. Hence
$\lambda+\lambda^q=t$ and $\lambda^{q+1}=1$.
\item For $u,v\in\mathbb F_q$, $\lambda(u+\lambda v)=-v+\lambda(u+tv)$.
Thus, in the ordered basis $(1,\lambda)$, multiplication by $\lambda$
sends the coordinate vector $(u,v)^t$ to $M_t(u,v)^t$.
\end{enumerate}
\end{lemma}

\begin{proof}
(i) Since $z^2-tz+1$ is irreducible, $\lambda\notin\mathbb F_q$, so $1$ and
$\lambda$ are linearly independent over $\mathbb F_q$; and
$[\mathbb F_{q^2}:\mathbb F_q]=2$.
(ii) The map $x\mapsto x^q$ permutes the roots of a polynomial over
$\mathbb F_q$ and fixes exactly the elements of $\mathbb F_q$. Since
$\lambda\notin\mathbb F_q$, the other root is $\lambda^q$. The sum and the
product of the roots are $t$ and $1$.
(iii) Since $\lambda$ is a root, $\lambda^2=t\lambda-1$. Hence
$\lambda(u+\lambda v)=\lambda u+v(t\lambda-1)=-v+\lambda(u+tv)$.
\end{proof}

We next introduce the Hermitian inner product. Let $N:\mathbb F_{q^2}\to\mathbb F_q$,
$N(a)=a\cdot a^q=a^{q+1}$, be the norm map; its values lie in
$\mathbb F_q$ because $N(a)^q=a^{q^2}a^q=N(a)$. For vectors
$a=(a_0,\ldots,a_{n-1})$ and $b=(b_0,\ldots,b_{n-1})$ in
$\mathbb F_{q^2}^n$, the Hermitian inner product is $\langle a,b\rangle_H=\sum_{i=0}^{n-1}a_ib_i^q.$ It is linear in the first argument and conjugate-linear in the second, so
$\langle\lambda a,\lambda b\rangle_H=\lambda\lambda^q\langle a,b\rangle_H
=\lambda^{q+1}\langle a,b\rangle_H$.

For a subset $E\subseteq\mathbb F_{q^2}^n$, the
Hermitian dual of $E$ is $E^{\perp_H}=\bigl\{a\in\mathbb F_{q^2}^n:\ \langle a,b\rangle_H=0
\ \text{for all } b\in E\bigr\},$
and $E$ is \emph{Hermitian self-orthogonal} if $E\subseteq E^{\perp_H}$.
The $\lambda$-constacyclic shift
$T_\lambda:\mathbb F_{q^2}^n\to\mathbb F_{q^2}^n$ is given by the same
formula as in Section~\ref{sec:shifts},
$T_\lambda(a_0,\ldots,a_{n-1})=(\lambda a_{n-1},a_0,\ldots,a_{n-2})$, now with $\lambda\in\mathbb F_{q^2}$. An $\mathbb F_{q^2}$-linear code
$E\subseteq\mathbb F_{q^2}^n$ is $\lambda$-constacyclic if
$T_\lambda(E)=E$. We write $w_H$ for the Hamming weight on
$\mathbb F_{q^2}^n$, defined as in Section~\ref{sec:preliminaries}.

Part~(iii) of Lemma~\ref{lambda-facts} explains the choice of the twist
matrix: $M_t$ is the matrix of multiplication by $\lambda$ on
$\mathbb F_{q^2}$ in the ordered basis $(1,\lambda)$. Its characteristic
polynomial is $\det(zI-M_t)=z(z-t)+1=z^2-tz+1$, whose roots are $\lambda$
and $\lambda^q$. Hence $\det M_t=\lambda\lambda^q=N(\lambda)$ and
$\operatorname{tr}M_t=\lambda+\lambda^q=t$. Since $\det M_t=1$,
\begin{equation}\label{eq:norm-one}
N(\lambda)=\lambda^{q+1}=1 .
\end{equation}
Thus the two conditions coincide: $\det M_t=1$ makes $S_{M_t}$ preserve
the symplectic product (Lemma~\ref{boundary-form}), and $N(\lambda)=1$
makes multiplication by $\lambda$ preserve the Hermitian product, since
$\langle\lambda a,\lambda b\rangle_H=N(\lambda)\langle a,b\rangle_H$.

By part~(iii), multiplying $u_i+\lambda v_i$ by $\lambda$ is the same as
multiplying the column $z_i=(u_i,v_i)^t$ by $M_t$. We therefore replace
each pair $(u_i,v_i)$ by the single element
$u_i+\lambda v_i\in\mathbb F_{q^2}$. The next proposition shows that this
identifies the matrix-twisted description over $\mathbb F_q$ with the
constacyclic description over $\mathbb F_{q^2}$.

\begin{proposition}\label{matrix-hermitian}
Suppose that $z^2-tz+1$ is irreducible over $\mathbb F_q$, and let
$\lambda\in\mathbb F_{q^2}$ be a root. Define
\[
\phi:\mathbb F_q^{2n}\longrightarrow\mathbb F_{q^2}^{\,n},\qquad
\phi(u\mid v)=(u_0+\lambda v_0,\ u_1+\lambda v_1,\ \ldots,\
u_{n-1}+\lambda v_{n-1}).
\]
Then the following hold.
\begin{enumerate}[label=\textup{(\roman*)}]
\item The map $\phi$ is an $\mathbb F_q$-linear bijection, and
$w_H(\phi(c))=w_s(c)$ for every $c\in\mathbb F_q^{2n}$.
\item The map $\phi$ carries the matrix-twisted shift to the
constacyclic shift: $\phi\circ S_{M_t}=T_\lambda\circ\phi$.
\item Let $C_{M_t}\subseteq\mathbb F_q^{2n}$ be an $M_t$-twisted code of
dimension $r$ over $\mathbb F_q$. Then $E=\phi(C_{M_t})$ is an
$\mathbb F_{q^2}$-linear $\lambda$-constacyclic code of length $n$ and
dimension $r/2$.
\item Let $c=(u\mid v)$ and $c'=(u'\mid v')$ be vectors of
$\mathbb F_q^{2n}$, and put $a=\phi(c)$ and $b=\phi(c')$. Then
\begin{equation}\label{eq:herm-symp}
\langle a,b\rangle_H-\langle a,b\rangle_H^{\,q}
=(\lambda^q-\lambda)\,\langle c,c'\rangle_s .
\end{equation}
\end{enumerate}
\end{proposition}

\begin{proof}
(i) This follows from Lemma~\ref{lambda-facts}(i): each pair
$(u_i,v_i)$ corresponds to a unique $u_i+\lambda v_i$, which is zero
exactly when the pair is zero.
(ii) Apply Lemma~\ref{lambda-facts}(iii) to the pair $z_{n-1}$, which is
the only pair multiplied by $M_t$ under the shift.
(iii) By (ii), $E$ is an $\mathbb F_q$-subspace invariant under
$T_\lambda$. Since $T_\lambda^{\,n}$ is multiplication by $\lambda$, $E$ is
closed under $\mathbb F_q+\lambda\mathbb F_q=\mathbb F_{q^2}$. Its
$\mathbb F_q$-dimension is $r$, so its $\mathbb F_{q^2}$-dimension is
$r/2$.
(iv) Write $a_i=u_i+\lambda v_i$ and $b_i=u_i'+\lambda v_i'$. Since
$x\mapsto x^q$ fixes $\mathbb F_q$ and $\lambda^{q+1}=1$
by~\eqref{eq:norm-one},
\[
a_ib_i^q=u_iu_i'+v_iv_i'+\lambda^q u_iv_i'+\lambda v_iu_i',
\]
and raising this to the $q$th power interchanges $\lambda$ and
$\lambda^q$. Subtracting gives
$a_ib_i^q-(a_ib_i^q)^q=(\lambda^q-\lambda)(u_iv_i'-v_iu_i')$; summing over
$i$ gives~\eqref{eq:herm-symp}.
\end{proof}

\begin{theorem}\label{hermitian-dual}
Suppose that $z^2-tz+1$ is irreducible over $\mathbb F_q$. Let
$C_{M_t}\subseteq\mathbb F_q^{2n}$ be an $M_t$-twisted code, and let
$E=\phi(C_{M_t})$, with $\phi$ as in Proposition~\ref{matrix-hermitian}.
Then $\phi\bigl(C_{M_t}^{\perp_s}\bigr)=E^{\perp_H}$.
Consequently, $C_{M_t}$ is symplectic self-orthogonal if and only if
$E$ is Hermitian self-orthogonal.
\end{theorem}

\begin{proof}
Since $\lambda\notin\mathbb F_q$, we have $\lambda^q-\lambda\ne0$.

Let $c\in C_{M_t}^{\perp_s}$ and put $a=\phi(c)$. Let
$b\in E$; then $b=\phi(c')$ for some $c'\in C_{M_t}$, and
$\langle c,c'\rangle_s=0$. Put $\eta=\langle a,b\rangle_H$. By
Proposition~\ref{matrix-hermitian}(iv),
$\eta-\eta^q=(\lambda^q-\lambda)\langle c,c'\rangle_s=0$, so
$\eta\in\mathbb F_q$. Since $E$ is $\mathbb F_{q^2}$-linear,
we have $\lambda b\in E$. Applying the same argument to $\lambda b$
gives $\langle a,\lambda b\rangle_H\in\mathbb F_q$. Moreover,
\[
\langle a,\lambda b\rangle_H
=\sum_ia_i(\lambda b_i)^q
=\lambda^q\sum_ia_ib_i^q
=\lambda^q\eta .
\]
If $\eta\ne0$,  then $\lambda^q=(\lambda^q\eta)\eta^{-1}\in\mathbb F_q$,
and hence $\lambda=(\lambda^q)^q\in\mathbb F_q$, a contradiction. So
$\eta=0$. Since $b\in E$ was arbitrary, $a\in E^{\perp_H}$. Thus
$\phi(C_{M_t}^{\perp_s})\subseteq E^{\perp_H}$.

Put $r=\dim_{\mathbb F_q}C_{M_t}$.
Both sides have $\mathbb F_q$-dimension $2n-r$. For the left side, this
holds because $\dim C_{M_t}^{\perp_s}=2n-r$ and $\phi$ is a bijection.
For the right side, $\dim_{\mathbb F_{q^2}}E=r/2$ by
Proposition~\ref{matrix-hermitian}(iii), so
$\dim_{\mathbb F_{q^2}}E^{\perp_H}=n-r/2$ because the Hermitian form is
nondegenerate. Hence $\phi(C_{M_t}^{\perp_s})=E^{\perp_H}$.

For the last assertion, since $\phi$ is a bijection,
\[
C_{M_t}\subseteq C_{M_t}^{\perp_s}
\iff
\phi(C_{M_t})\subseteq\phi(C_{M_t}^{\perp_s})
\iff
E\subseteq E^{\perp_H}. \qedhere
\]
\end{proof}

Proposition~\ref{matrix-hermitian} and Theorem~\ref{hermitian-dual} show
that the matrix-twisted description over $\mathbb F_q$ and the
constacyclic description over $\mathbb F_{q^2}$ are the same
construction written in different coordinates. The map $\phi$ turns the
shift $S_{M_t}$ into $T_\lambda$, preserves weights, and sends the
symplectic dual to the Hermitian dual. In particular, $\phi$ maps
$C_{M_t}^{\perp_s}\setminus C_{M_t}$ bijectively onto
$E^{\perp_H}\setminus E$ without changing weights. When $C_{M_t}$ is symplectic self orthogonal, the two descriptions give quantum codes with the same parameters.
For $k>0$, the distance is
$\min\{w_H(a):a\in E^{\perp_H}\setminus E\}$.
For $k=0$, it is $d_H(E)=d_s(C_{M_t})$.

Using this construction, we obtain the new quantum codes H1--H13 listed in
Table~\ref{tab:qt-direct}. For each of them, the twist $\lambda$ has norm
one, and the polynomials satisfy $gh=m_t$ and $h^*\mid g$, so
$C_{M_t}=\langle g\rangle$ is symplectic self-orthogonal. By
Lemma~\ref{matrix-ideal-basis}, its $\mathbb F_q$-dimension is
$r=\deg h=2n-\deg g$, and Theorem~\ref{Q} gives a quantum code with
dimension $k=n-r$; the column $r$ of Table~\ref{tab:qt-direct} lists
$r=\dim_{\mathbb F_q}C_{M_t}$. Table~\ref{tab:twists} gives the minimal polynomials and orders
of the twists, and Table~\ref{tab:qt-generators} gives the coefficients of
$g$ and $h$.

\begin{table}[H]
\centering\small\twocolsetup
\setlength{\tabcolsep}{4pt}
\caption{New quantum codes from Hermitian constacyclic codes. }\label{tab:qt-direct}
\begin{tabular}{@{}llccc@{\hspace{1.2em}}V@{\hspace{1.2em}}llccc@{}}
\toprule
No. & Quantum code & Twist & $r$ & Previous~\cite{gra} &
No. & Quantum code & Twist & $r$ & Previous~\cite{gra}\\
\midrule
H1 & $[[91,73,6]]_3$  & $M_0$ & 18 & 5--7   & H8  & $[[52,16,13]]_5$ & $M_1$ & 36 & 12--18\\
H2 & $[[52,32,8]]_5$  & $M_1$ & 20 & 7--10  & H9  & $[[52,14,14]]_5$ & $M_4$ & 38 & 13--19\\
H3 & $[[50,28,9]]_7$  & $M_3$ & 22 & 8--12  & H10 & $[[75,45,11]]_7$ & $M_3$ & 30 & 10--15\\
H4 & $[[42,6,14]]_5$  & $M_1$ & 36 & 13--19 & H11 & $[[50,12,15]]_7$ & $M_4$ & 38 & 14--20\\
H5 & $[[50,20,12]]_7$ & $M_3$ & 30 & 11--16 & H12 & $[[52,12,15]]_5$ & $M_1$ & 40 & 13--20\\
H6 & $[[91,61,9]]_3$  & $M_0$ & 30 & 8--13  & H13 & $[[50,10,16]]_7$ & $M_0$ & 40 & 14--21\\
H7 & $[[52,20,12]]_5$ & $M_1$ & 32 & 11--16 &  &  &  &  & \\
\bottomrule
\end{tabular}
\end{table}

\begin{table}[H]
\centering\small\twocolsetup
\setlength{\tabcolsep}{5pt}
\caption{The twists $M_t$ used in Table~\ref{tab:qt-direct}, with the
minimal polynomial and order of the corresponding
$\lambda\in\mathbb F_{q^2}$.}\label{tab:twists}
\begin{tabular}{@{}cccc@{\hspace{1.2em}}V@{\hspace{1.2em}}cccc@{}}
\toprule
$q$ & $t$ & Min.\ polynomial & $\operatorname{ord}(\lambda)$ &
$q$ & $t$ & Min.\ polynomial & $\operatorname{ord}(\lambda)$\\
\midrule
$3$ & $0$ & $z^2+1$   & $4$ & $7$ & $0$ & $z^2+1$    & $4$\\
$5$ & $1$ & $z^2-z+1$ & $6$ & $7$ & $3$ & $z^2-3z+1$ & $8$\\
$5$ & $4$ & $z^2+z+1$ & $3$ & $7$ & $4$ & $z^2+3z+1$ & $8$\\
\bottomrule
\end{tabular}
\end{table}

\subsection{A rank criterion for the quantum minimum distance}
For $k>0$, the quantum distance is
$d=\min\{w_s(c):c\in C^{\perp_s}\setminus C\}$. Since a code may be
impure, computing $d$ requires excluding the vectors of $C$ from
$C^{\perp_s}$. The following criterion does this using only ranks of
submatrices of a generator matrix of $C$. Delfosse and
Z\'emor~\cite{DZ2013} give the binary version;
we state the $\mathbb F_q$-linear form in our coordinates.

\begin{proposition}\label{support-test}
Let $C\subseteq C^{\perp_s}\subseteq\mathbb F_q^{2n}$ have dimension $r$,
and let $G$ be an $r\times 2n$ generator matrix of $C$. For
$I\subseteq\{0,\ldots,n-1\}$, put
\[
V_I=\{(u\mid v)\in\mathbb F_q^{2n}:(u_i,v_i)=(0,0)\text{ for }i\notin I\},
\]
let $G_I$ consist of the columns $i$ and $n+i$ of $G$ for $i\in I$, and
define $G_{I^c}$ in the same way. Then
\[
\dim(C^{\perp_s}\cap V_I)-\dim(C\cap V_I)
=2|I|-\operatorname{rank}G_I-r+\operatorname{rank}G_{I^c}.
\]
In particular, $C^{\perp_s}\setminus C$ contains a vector supported in $I$
if and only if the right-hand side is positive.
\end{proposition}

\begin{proof}
For $x\in V_I$, let $x_I\in\mathbb F_q^{2|I|}$ be its restriction to the
sites in $I$. Then $x\in C^{\perp_s}$ if and only if
$x_I\Omega_IG_I^t=0$, where $\Omega_I$ is the (invertible) symplectic
matrix on these sites. Since $\operatorname{rank}(\Omega_IG_I^t)
=\operatorname{rank}G_I$, we get
$\dim(C^{\perp_s}\cap V_I)=2|I|-\operatorname{rank}G_I$.

A codeword $aG$, with $a\in\mathbb F_q^r$, lies in $V_I$ exactly when
$aG_{I^c}=0$. Since the rows of $G$ are independent, $a\mapsto aG$ is
injective, so $\dim(C\cap V_I)=r-\operatorname{rank}G_{I^c}$. Subtracting
gives the identity.

Finally, $C\cap V_I\subseteq C^{\perp_s}\cap V_I$, so the difference is
nonnegative, and it is positive exactly when some vector of
$C^{\perp_s}\cap V_I$ does not lie in $C$.
\end{proof}

By Proposition~\ref{support-test}, the distance is at least $d$ if the
right-hand side vanishes for every $I$ with $|I|<d$, and it is at most $d$
if $C^{\perp_s}\setminus C$ contains a vector of symplectic weight $d$. For
a self-dual code such as O1, the distance is $d_s(C)$ instead: it is the
smallest $|I|$ with $C\cap V_I\ne0$, that is, with
$\operatorname{rank}G_{I^c}<r$.

\subsection{Further parameters from propagation}\label{sec:propagation}

We obtain further parameters from the explicit codes by three standard
propagation rules for stabilizer codes.

\begin{proposition}[{\cite{Ket}}]\label{prop-rules}
Suppose that an $[[n,k,d]]_q$ quantum code is obtained from
Theorem~\ref{Q}.
\begin{enumerate}[label=\textup{(\roman*)}]
\item \emph{Dimension reduction.} If $k\ge2$, there is an
$[[n,k-1,\ge d]]_q$ code.
\item \emph{Lengthening.} If $k\ge1$, there is an $[[n+1,k,d]]_q$ code;
it is impure when $d>1$.
\item \emph{Puncturing.} If the code is pure and $d\ge2$, there
is a pure $[[n-1,k+1,d-1]]_q$ code.
\end{enumerate}
\end{proposition}

We apply these rules in a fixed order. For a parent $[[n,k,d]]_q$, the
triple $(a,s,\ell)$ means $a$ puncturings, followed by $s$ dimension
reductions and $\ell$ lengthenings; this gives
\[
[[\,n-a+\ell,\;k+a-s,\;\ge d-a\,]]_q .
\]
We puncture only pure parents, and only before lengthening, since a
lengthened code is impure. When $a=0$ we write $(s,\ell)$. All target
dimensions in our tables are positive, so every dimension reduction is
valid. Table~\ref{tab:rules} summarizes the rules and where we use them.

\begin{table}[H]
\centering\small
\setlength{\tabcolsep}{5pt}
\caption{The propagation rules of Proposition~\ref{prop-rules} and where
they are used.}\label{tab:rules}
\begin{tabular}{@{}llll@{}}
\toprule
Rule & Result from $[[n,k,d]]_q$ & Condition & Used in\\
\midrule
Dimension reduction & $[[n,k-1,\ge d]]_q$ & $k\ge2$ &
Tables~\ref{tab:dominated}--\ref{tab:pure-parent-derived}\\
Lengthening & $[[n+1,k,d]]_q$ & $k\ge1$ &
Tables~\ref{tab:qc-derived}--\ref{tab:pure-parent-derived}\\
Puncturing & $[[n-1,k+1,d-1]]_q$, pure & pure, $d\ge2$ &
Table~\ref{tab:pure-parent-derived} (H12, H13)\\
\bottomrule
\end{tabular}
\end{table}

\paragraph{Explicit codes implied by other explicit codes.}
Dimension reduction also relates the explicit codes among themselves.
Nine of the codes N1--N44 have parameters that follow from another
explicit code by Proposition~\ref{prop-rules}(i); Table~\ref{tab:dominated}
lists them, and Table~\ref{tab:qc-direct} marks them with a dagger. We count these nine entries once among the explicit constructions and do not count them again through propagation. We keep them because their
generators are given explicitly in Table~\ref{tab:qc-generators}, rather
than only being implied by another code.

\begin{table}[H]
\centering\small\twocolsetup
\setlength{\tabcolsep}{4pt}
\caption{Explicit codes whose parameters follow from another explicit
code by $s$ dimension reductions. The column ``Bound'' is the resulting
distance lower bound.}\label{tab:dominated}
\begin{tabular}{@{}llcc@{\hspace{0.9em}}V@{\hspace{0.9em}}llcc@{\hspace{0.9em}}V@{\hspace{0.9em}}llcc@{}}
\toprule
Code & Parent & $s$ & Bound & Code & Parent & $s$ & Bound & Code & Parent & $s$ & Bound\\
\midrule
N9  & N8  & $1$ & $\ge10$ & N15 & N17 & $2$ & $\ge7$  & N39 & H13 & $1$ & $\ge16$\\
N10 & N11 & $1$ & $\ge7$  & N16 & N17 & $1$ & $\ge7$  & N40 & H12 & $2$ & $\ge15$\\
N12 & N13 & $1$ & $\ge9$  & N32 & N33 & $1$ & $\ge11$ & N44 & H13 & $2$ & $\ge16$\\
\bottomrule
\end{tabular}
\end{table}

\paragraph{Derived codes.}
Applying dimension reduction and lengthening to the codes N1--N44 gives
the $29$ parameter sets in Table~\ref{tab:qc-derived}, and applying them to T1--T4 and H1--H11 gives the $58$ parameter sets in
Table~\ref{tab:qt-derived}. The explicit codes N1--N44, T1--T4, and
H1--H11, together with these two tables, give $146$ distinct cells.
Section~\ref{sec:pure-propagation} adds the parameters from H12 and H13.

\begin{table}[H]
\centering\footnotesize\twocolsetup
\setlength{\tabcolsep}{2.2pt}
\caption{Parameter sets derived from N1--N44 by $(s,\ell)$. Previous
bounds are from~\cite{gra}.}\label{tab:qc-derived}
\begin{tabular}{@{}lllcc@{\hspace{0.9em}}V@{\hspace{0.9em}}lllcc@{}}
\toprule
No. & Quantum parameters & Parent & $(s,\ell)$ & Previous & No. & Quantum parameters & Parent & $(s,\ell)$ & Previous\\
\midrule
D1 & $[[31,14,\ge 7]]_5$ & N2 & (1,0) & 6--9 & D16 & $[[95,68,\ge 8]]_5$ & N14 & (0,2) & 7--13\\
D2 & $[[71,44,\ge 9]]_5$ & N6 & (1,0) & 8--14 & D17 & $[[94,71,\ge 7]]_5$ & N17 & (1,1) & 6--11\\
D3 & $[[71,43,\ge 9]]_5$ & N6 & (2,0) & 8--14 & D18 & $[[94,72,\ge 7]]_5$ & N17 & (0,1) & 6--11\\
D4 & $[[71,42,\ge 9]]_5$ & N6 & (3,0) & 8--15 & D19 & $[[95,72,\ge 7]]_5$ & N17 & (0,2) & 6--11\\
D5 & $[[72,45,\ge 9]]_5$ & N6 & (0,1) & 8--14 & D20 & $[[41,13,\ge 11]]_7$ & N27 & (0,1) & 10--15\\
D6 & $[[73,45,\ge 9]]_5$ & N6 & (0,2) & 8--14 & D21 & $[[57,24,\ge 12]]_7$ & N34 & (1,0) & 11--17\\
D7 & $[[74,45,\ge 9]]_5$ & N6 & (0,3) & 8--15 & D22 & $[[58,25,\ge 12]]_7$ & N34 & (0,1) & 11--17\\
D8 & $[[72,51,\ge 7]]_5$ & N4 & (0,1) & 6--11 & D23 & $[[48,5,\ge 15]]_5$ & N37 & (1,0) & 14--22\\
D9 & $[[79,51,\ge 9]]_5$ & N7 & (0,1) & 8--14 & D24 & $[[49,6,\ge 15]]_5$ & N37 & (0,1) & 14--22\\
D10 & $[[80,51,\ge 9]]_5$ & N7 & (0,2) & 8--15 & D25 & $[[49,3,\ge 16]]_5$ & N41 & (0,1) & 15--23\\
D11 & $[[72,41,\ge 10]]_5$ & N8 & (0,1) & 9--16 & D26 & $[[48,5,\ge 16]]_7$ & N42 & (1,0) & 15--22\\
D12 & $[[92,70,\ge 7]]_5$ & N11 & (1,1) & 6--11 & D27 & $[[49,6,\ge 16]]_7$ & N42 & (0,1) & 15--22\\
D13 & $[[92,71,\ge 7]]_5$ & N11 & (0,1) & 6--10 & D28 & $[[57,18,\ge 14]]_7$ & N43 & (1,0) & 13--20\\
D14 & $[[94,65,\ge 9]]_5$ & N13 & (0,1) & 8--14 & D29 & $[[58,19,\ge 14]]_7$ & N43 & (0,1) & 13--20\\
D15 & $[[94,68,\ge 8]]_5$ & N14 & (0,1) & 7--13 &  &  &  &  & \\
\bottomrule
\end{tabular}
\end{table}
\begin{table}[H]
\centering\footnotesize\twocolsetup
\setlength{\tabcolsep}{2.0pt}
\caption{Parameter sets derived from T1--T4 and H1--H11 by $(s,\ell)$.
Overlaps with H12 and H13 appear in Table~\ref{tab:pure-parent-derived}.
Previous bounds are from~\cite{gra}.}\label{tab:qt-derived}
\begin{tabular}{@{}lllcc@{\hspace{0.6em}}V@{\hspace{0.6em}}lllcc@{}}
\toprule
No. & Quantum parameters & Parent & $(s,\ell)$ & Previous & No. & Quantum parameters & Parent & $(s,\ell)$ & Previous\\
\midrule
T1a & $[[84,61,\ge 7]]_5$ & T1 & (1,0) & 6--11 & H7i & $[[55,20,\ge 12]]_5$ & H7 & (0,3) & 11--18\\
T1b & $[[85,62,\ge 7]]_5$ & T1 & (0,1) & 6--11 & H8a & $[[53,16,\ge 13]]_5$ & H8 & (0,1) & 12--19\\
T2a & $[[61,36,\ge 9]]_7$ & T2 & (0,1) & 8--13 & H8b & $[[52,15,\ge 13]]_5$ & H8 & (1,0) & 12--19\\
T3a & $[[51,16,\ge 13]]_7$ & T3 & (0,1) & 12--18 & H8c & $[[53,15,\ge 13]]_5$ & H8 & (1,1) & 12--19\\
T4a & $[[53,8,\ge 14]]_3$ & T4 & (0,1) & 13--20 & H8d & $[[54,15,\ge 13]]_5$ & H8 & (1,2) & 12--20\\
H1a & $[[91,72,\ge 6]]_3$ & H1 & (1,0) & 5--8 & H8e & $[[54,16,\ge 13]]_5$ & H8 & (0,2) & 12--19\\
H1b & $[[92,73,\ge 6]]_3$ & H1 & (0,1) & 5--8 & H8f & $[[55,16,\ge 13]]_5$ & H8 & (0,3) & 12--20\\
H3a & $[[51,28,\ge 9]]_7$ & H3 & (0,1) & 8--12 & H9a & $[[52,13,\ge 14]]_5$ & H9 & (1,0) & 13--20\\
H4a & $[[43,6,\ge 14]]_5$ & H4 & (0,1) & 13--19 & H9b & $[[52,11,\ge 14]]_5$ & H9 & (3,0) & 13--21\\
H4b & $[[42,5,\ge 14]]_5$ & H4 & (1,0) & 13--19 & H9c & $[[53,14,\ge 14]]_5$ & H9 & (0,1) & 13--20\\
H4c & $[[44,6,\ge 14]]_5$ & H4 & (0,2) & 13--19 & H9d & $[[53,13,\ge 14]]_5$ & H9 & (1,1) & 13--20\\
H4d & $[[45,6,\ge 14]]_5$ & H4 & (0,3) & 13--20 & H9e & $[[53,12,\ge 14]]_5$ & H9 & (2,1) & 13--21\\
H4e & $[[46,6,\ge 14]]_5$ & H4 & (0,4) & 13--20 & H9f & $[[54,14,\ge 14]]_5$ & H9 & (0,2) & 13--20\\
H5a & $[[51,20,\ge 12]]_7$ & H5 & (0,1) & 11--16 & H9g & $[[54,13,\ge 14]]_5$ & H9 & (1,2) & 13--21\\
H5b & $[[50,19,\ge 12]]_7$ & H5 & (1,0) & 11--16 & H9h & $[[55,14,\ge 14]]_5$ & H9 & (0,3) & 13--21\\
H5c & $[[51,19,\ge 12]]_7$ & H5 & (1,1) & 11--17 & H10a & $[[75,44,\ge 11]]_7$ & H10 & (1,0) & 10--16\\
H5d & $[[52,19,\ge 12]]_7$ & H5 & (1,2) & 11--17 & H10b & $[[76,45,\ge 11]]_7$ & H10 & (0,1) & 10--16\\
H5e & $[[52,20,\ge 12]]_7$ & H5 & (0,2) & 11--17 & H10c & $[[76,44,\ge 11]]_7$ & H10 & (1,1) & 10--16\\
H5f & $[[53,20,\ge 12]]_7$ & H5 & (0,3) & 11--17 & H10d & $[[77,45,\ge 11]]_7$ & H10 & (0,2) & 10--16\\
H6a & $[[91,60,\ge 9]]_3$ & H6 & (1,0) & 8--13 & H10e & $[[77,44,\ge 11]]_7$ & H10 & (1,2) & 10--17\\
H6b & $[[92,61,\ge 9]]_3$ & H6 & (0,1) & 8--13 & H10f & $[[78,45,\ge 11]]_7$ & H10 & (0,3) & 10--17\\
H7a & $[[53,20,\ge 12]]_5$ & H7 & (0,1) & 11--17 & H11a & $[[50,11,\ge 15]]_7$ & H11 & (1,0) & 14--20\\
H7b & $[[52,17,\ge 12]]_5$ & H7 & (3,0) & 11--18 & H11b & $[[50,10,\ge 15]]_7$ & H11 & (2,0) & 14--21\\
H7c & $[[52,18,\ge 12]]_5$ & H7 & (2,0) & 11--17 & H11c & $[[51,12,\ge 15]]_7$ & H11 & (0,1) & 14--20\\
H7d & $[[52,19,\ge 12]]_5$ & H7 & (1,0) & 11--17 & H11d & $[[51,11,\ge 15]]_7$ & H11 & (1,1) & 14--21\\
H7e & $[[53,18,\ge 12]]_5$ & H7 & (2,1) & 11--18 & H11e & $[[51,10,\ge 15]]_7$ & H11 & (2,1) & 14--21\\
H7f & $[[53,19,\ge 12]]_5$ & H7 & (1,1) & 11--17 & H11f & $[[52,12,\ge 15]]_7$ & H11 & (0,2) & 14--21\\
H7g & $[[54,19,\ge 12]]_5$ & H7 & (1,2) & 11--18 & H11g & $[[52,11,\ge 15]]_7$ & H11 & (1,2) & 14--21\\
H7h & $[[54,20,\ge 12]]_5$ & H7 & (0,2) & 11--17 & H11h & $[[53,12,\ge 15]]_7$ & H11 & (0,3) & 14--21\\
\bottomrule
\end{tabular}
\end{table}

\subsection{Propagation from two pure parents}\label{sec:pure-propagation}

For H12 $=[[52,12,15]]_5$ and H13 $=[[50,10,16]]_7$, the polynomials $g$
and $h$ are listed in Table~\ref{tab:qt-generators}; in both cases
$\deg h=40$. Hence $C_{M_t}$ has $\mathbb F_q$-dimension $40$, the associated
quantum code has $k=n-40$, and the code $E=\phi(C_{M_t})$ over
$\mathbb F_{q^2}$ has dimension $20$ by
Proposition~\ref{matrix-hermitian}(iii). The codes $E$ and their
Hermitian duals have parameters
\[
\begin{aligned}
E_{\mathrm{H12}}&:[52,20,\ge15]_{25},&
E_{\mathrm{H12}}^{\perp_H}&:[52,32,15]_{25},\\
E_{\mathrm{H13}}&:[50,20,\ge16]_{49},&
E_{\mathrm{H13}}^{\perp_H}&:[50,30,16]_{49}.
\end{aligned}
\]
Since $\phi$ preserves weights, the minimum distance of $E$ equals
$d_s(C_{M_t})$. The quantum distances are $d=15$ for H12 and $d=16$ for
H13. We have $d_s(C_{M_t})\ge d$ in both cases.
Hence both quantum codes are
pure, and puncturing, Proposition~\ref{prop-rules}(iii), applies to them.

Table~\ref{tab:pure-parent-derived} lists the codes obtained from H12 and
H13 by $(a,s,\ell)$; for the parent rows we keep the exact distance. For
example, two puncturings of H13 give
\[
[[50,10,16]]_7\longrightarrow[[49,11,15]]_7\longrightarrow[[48,12,14]]_7 ;
\]
the second puncturing is valid because the first punctured code is again
pure. The recorded bounds for $[[48,12]]_7$ are $13$--$19$, so this gives
one more improved cell.

Table~\ref{tab:pure-parent-derived} has $30$ entries, $15$ from each parent. Of these, $17$ give new
cells, $5$ strengthen distances of entries listed earlier, and $8$ repeat
earlier entries with the same bound.

\begingroup\small
\setlength{\tabcolsep}{4pt}
\renewcommand{\arraystretch}{1.12}
\begin{longtable}{@{}llcll@{}}
\caption{The $30$ entries obtained from H12 and H13 by $(a,s,\ell)$. The
last column compares each entry with the entries listed
earlier.}\label{tab:pure-parent-derived}\\
\toprule
Quantum parameters & Parent & $(a,s,\ell)$ & Previous~\cite{gra} & Other entry\\\midrule
\endfirsthead
\multicolumn{5}{c}{\tablename~\thetable\ (continued)}\\\toprule
Quantum parameters & Parent & $(a,s,\ell)$ & Previous~\cite{gra} & Other entry\\\midrule
\endhead
\bottomrule\endfoot
\([[51,11,\ge 14]]_5\) & H12 & (1,2,0) & 13--20 & new cell\\
\([[51,12,\ge 14]]_5\) & H12 & (1,1,0) & 13--20 & new cell\\
\([[51,13,\ge 14]]_5\) & H12 & (1,0,0) & 13--19 & new cell\\
\([[52,9,\ge 15]]_5\) & H12 & (0,3,0) & 14--22 & new cell\\
\([[52,10,\ge 15]]_5\) & H12 & (0,2,0) & 14--21 & N40: same\\
\([[52,11,\ge 15]]_5\) & H12 & (0,1,0) & 13--21 & H9b: stronger\\
\([[52,12,15]]_5\) & H12 & (0,0,0) & 13--20 & new cell\\
\([[52,13,\ge 14]]_5\) & H12 & (1,0,1) & 13--20 & H9a: same\\
\([[53,10,\ge 15]]_5\) & H12 & (0,2,1) & 14--22 & new cell\\
\([[53,11,\ge 15]]_5\) & H12 & (0,1,1) & 14--21 & new cell\\
\([[53,12,\ge 15]]_5\) & H12 & (0,0,1) & 13--21 & H9e: stronger\\
\([[53,13,\ge 14]]_5\) & H12 & (1,0,2) & 13--20 & H9d: same\\
\([[54,11,\ge 15]]_5\) & H12 & (0,1,2) & 14--22 & new cell\\
\([[54,12,\ge 15]]_5\) & H12 & (0,0,2) & 14--21 & new cell\\
\([[54,13,\ge 14]]_5\) & H12 & (1,0,3) & 13--21 & H9g: same\\
\([[48,12,\ge 14]]_7\) & H13 & (2,0,0) & 13--19 & new cell\\
\([[49,10,\ge 15]]_7\) & H13 & (1,1,0) & 14--20 & new cell\\
\([[49,11,\ge 15]]_7\) & H13 & (1,0,0) & 14--20 & new cell\\
\([[50,7,\ge 16]]_7\) & H13 & (0,3,0) & 15--22 & new cell\\
\([[50,8,\ge 16]]_7\) & H13 & (0,2,0) & 15--22 & N44: same\\
\([[50,9,\ge 16]]_7\) & H13 & (0,1,0) & 14--21 & N39: stronger\\
\([[50,10,16]]_7\) & H13 & (0,0,0) & 14--21 & H11b: stronger\\
\([[50,11,\ge 15]]_7\) & H13 & (1,0,1) & 14--20 & H11a: same\\
\([[51,8,\ge 16]]_7\) & H13 & (0,2,1) & 15--22 & new cell\\
\([[51,9,\ge 16]]_7\) & H13 & (0,1,1) & 15--22 & new cell\\
\([[51,10,\ge 16]]_7\) & H13 & (0,0,1) & 14--21 & H11e: stronger\\
\([[51,11,\ge 15]]_7\) & H13 & (1,0,2) & 14--21 & H11d: same\\
\([[52,9,\ge 16]]_7\) & H13 & (0,1,2) & 15--22 & new cell\\
\([[52,10,\ge 16]]_7\) & H13 & (0,0,2) & 15--22 & new cell\\
\([[52,11,\ge 15]]_7\) & H13 & (1,0,3) & 14--21 & H11g: same\\
\end{longtable}
\endgroup

In summary, the codes N1--N44, T1--T4, and H1--H11 give $44+4+11=59$
cells directly. Propagation adds $29$ cells from N1--N44
(Table~\ref{tab:qc-derived}) and $58$ cells from T1--T4 and H1--H11
(Table~\ref{tab:qt-derived}), giving $59+29+58=146$ cells. The pure
parents H12 and H13 add $17$ further cells
(Table~\ref{tab:pure-parent-derived}), giving $146+17=163$, and the
example O6 of Table~\ref{tab:qc-examples} adds one more, for $164$ in
total.

The nine daggered codes of Table~\ref{tab:dominated} could also be
obtained from other explicit codes by dimension reduction. Since we give
their generators explicitly, we count their cells only once, among the
$59$ explicit codes, and do not list them again in the propagation
tables.

\section{Conclusion}\label{sec:conclusion}
For index-$2$ quasi-twisted codes, a scalar twist acting on both blocks gives
quantum distance greater than $1$ only for $\lambda=\pm1$.
A determinant-one matrix twist preserves the symplectic form, and an irreducible characteristic polynomial gives the norm-one Hermitian constacyclic construction. Writing these codes over the base
field keeps their generators, duals, and self-orthogonality test in the
same coordinates as the quantum construction. Our constructions over $\mathbb F_3$, $\mathbb F_5$, and $\mathbb F_7$, together with the propagation rules, yield quantum codes that improve the recorded minimum-distance lower bounds for $164$ distinct entries in Grassl's tables.

Two directions remain open. First, the matrix-twist description over
$\mathbb F_q$ does not require $z^2-tz+1$ to be irreducible; when it splits, the quadratic-extension identification above no longer applies, and it would be
interesting to see whether they give further good codes. Second, the same approach could be applied to codes of larger index.

\section*{Acknowledgement}
Part of this work was carried out while T. Bag was with Inria, ENS de Lyon,
and supported by the European Research Council (ERC Grant AlgoQIP, Agreement
No.\ 851716) and by the Agence Nationale de la Recherche (ANR) under the
France 2030 program, grant ANR-22-PETQ-0006. D. Panario was partially funded by the
Natural Sciences and Engineering Research Council of Canada (NSERC),
reference number RGPIN-2024-05341.

\appendix
\section{\texorpdfstring{Generator coefficients for the QC examples}{Generator coefficients for the QC examples}}

We write $f(x)=c_d x^d+\cdots+c_1x+c_0$ as the coefficient string $c_dc_{d-1}\cdots c_0$, including every zero coefficient.  Long strings are broken across lines only for typesetting; the string
continues directly on the next line. 

\begin{table}[H]
\centering\footnotesize\twocolsetup
\setlength{\tabcolsep}{3pt}
\caption{Generator coefficients for Table~\ref{tab:qc-examples}, with
modulus $x^n-1$.}\label{tab:qc-example-generators}
\begin{tabular}{@{}p{0.035\textwidth}p{0.24\textwidth}@{\hspace{0.6em}}V@{\hspace{0.6em}}p{0.035\textwidth}p{0.24\textwidth}@{\hspace{0.6em}}V@{\hspace{0.6em}}p{0.035\textwidth}p{0.28\textwidth}@{}}
\toprule
No. & Polynomials & No. & Polynomials & No. & Polynomials\\
\midrule
O1 & $g$: \coef{1}\newline $r_1$: \coef{444444444}\newline $r_2$: \coef{4242400} &
O3 & $g$: \coef{1202121101001}\newline $r_1$: \coef{222222222222}\newline $r_2$: \coef{22222221201} &
O5 & $g$: \coef{1221}\newline $r_1$: \coef{11}\newline $r_2$: \coef{11021220112212112202}\\[5pt]
O2 & $g$: \coef{2121112}\newline $r_1$: \coef{2222221}\newline $r_2$: \coef{2221222} &
O4 & $g$: \coef{2010001}\newline $r_1$: \coef{2222222221}\newline $r_2$: \coef{2212112010} &
O6 & $g$: \coef{12122221121022}\newline $r_1$: \coef{1}\newline $r_2$: \coef{2211211010020021100102}\\
\bottomrule
\end{tabular}
\end{table}

\section{\texorpdfstring{Generators for the scalar-twist codes}{Generators for the scalar-twist codes}}
For T1--T4 and N1--N44, $r_1=1$ and $C=\langle(g,r_2g)\rangle$ in
$R_\lambda^2$, with $\lambda=-1$ for T1--T4 and $\lambda=1$ for N1--N44.
Tables~\ref{tab:negacyclic} and~\ref{tab:qc-direct} supply $q$ and $n$,
and the degree of $g$ is the quantum dimension $k$.

\begingroup
\small
\setlength{\tabcolsep}{5pt}
\renewcommand{\arraystretch}{1.12}
\begin{longtable}{@{}p{0.05\textwidth}>{\raggedright\arraybackslash}p{0.89\textwidth}@{}}
\caption{Generator coefficients for T1--T4, with modulus $x^n+1$. In every
row $r_1=1$.}\label{tab:qn-generators}\\
\toprule
No. & Polynomials\\
\midrule
\endfirsthead
\bottomrule
\endlastfoot
T1 & \(g\) (degree $62$): \coef{1420130100214023\allowbreak{}1332404042124030\allowbreak{}3023013100142122\allowbreak{}302421000413112}\newline \(r_2\) (degree $21$): \coef{1410440421340203\allowbreak{}240101}\\[4pt]
T2 & \(g\) (degree $36$): \coef{1011436422506440\allowbreak{}1505055621143501\allowbreak{}22112}\newline \(r_2\) (degree $23$): \coef{6333110404353665\allowbreak{}10300123}\\[4pt]
T3 & \(g\) (degree $16$): \coef{1023622156256420\allowbreak{}1}\newline \(r_2\) (degree $33$): \coef{2263151500535366\allowbreak{}5566543216146255\allowbreak{}16}\\[4pt]
T4 & \(g\) (degree $8$): \coef{100122221}\newline \(r_2\) (degree $43$): \coef{1121222210100111\allowbreak{}0112011100101101\allowbreak{}211020201002}\\
\end{longtable}
\endgroup
\begingroup
\small
\setlength{\tabcolsep}{5pt}
\renewcommand{\arraystretch}{1.12}
\begin{longtable}{@{}p{0.05\textwidth}>{\raggedright\arraybackslash}p{0.89\textwidth}@{}}
\caption{Generator coefficients for N1--N44, with modulus $x^n-1$. In every row $r_1=1$.}\label{tab:qc-generators}\\
\toprule
No. & Polynomials\\
\midrule
\endfirsthead
\multicolumn{2}{c}{\tablename~\thetable\ (continued)}\\
\toprule
No. & Polynomials\\
\midrule
\endhead
\midrule
\multicolumn{2}{r}{Continued on the next page}\\
\endfoot
\bottomrule
\endlastfoot
N1 & \(g\) (degree $31$): \coef{1221212210110021\allowbreak{}0200022222010001}\newline \(r_2\) (degree $19$): \coef{1120200112210010\allowbreak{}2222}\\[4pt]
N2 & \(g\) (degree $15$): \coef{1212422244102424}\newline \(r_2\) (degree $15$): \coef{4101030130124434}\\[4pt]
N3 & \(g\) (degree $12$): \coef{1231430224031}\newline \(r_2\) (degree $18$): \coef{2333104000001002\allowbreak{}331}\\[4pt]
N4 & \(g\) (degree $51$): \coef{1123341203002223\allowbreak{}2223130234343223\allowbreak{}2321140303043413\allowbreak{}3244}\newline \(r_2\) (degree $19$): \coef{2121013242143343\allowbreak{}0041}\\[4pt]
N5 & \(g\) (degree $46$): \coef{1043440144044410\allowbreak{}4431323234013413\allowbreak{}000322422002211}\newline \(r_2\) (degree $24$): \coef{1030303304143440\allowbreak{}024100300}\\[4pt]
N6 & \(g\) (degree $45$): \coef{1333423401112323\allowbreak{}2102004213211444\allowbreak{}12323020200014}\newline \(r_2\) (degree $25$): \coef{1213331002024304\allowbreak{}3022434331}\\[4pt]
N7 & \(g\) (degree $51$): \coef{1330104144231233\allowbreak{}1040414233210232\allowbreak{}0240101101031221\allowbreak{}4244}\newline \(r_2\) (degree $26$): \coef{1024443004341041\allowbreak{}23430200312}\\[4pt]
N8 & \(g\) (degree $41$): \coef{1201044331342300\allowbreak{}2014041240134012\allowbreak{}2341433014}\newline \(r_2\) (degree $29$): \coef{3003040413121113\allowbreak{}34403022333241}\\[4pt]
N9 & \(g\) (degree $40$): \coef{1213310332310402\allowbreak{}4234141001113421\allowbreak{}202434311}\newline \(r_2\) (degree $30$): \coef{2434243131241130\allowbreak{}003224133200204}\\[4pt]
N10 & \(g\) (degree $70$): \coef{1341222101432322\allowbreak{}4033210332213041\allowbreak{}4040342241132404\allowbreak{}0004432214304333\allowbreak{}4001211}\newline \(r_2\) (degree $19$): \coef{3101122421044222\allowbreak{}2133}\\[4pt]
N11 & \(g\) (degree $71$): \coef{1212333013430244\allowbreak{}2112313132113202\allowbreak{}4340131344421434\allowbreak{}3320010414340122\allowbreak{}01402414}\newline \(r_2\) (degree $19$): \coef{4440024120431011\allowbreak{}1221}\\[4pt]
N12 & \(g\) (degree $64$): \coef{1001243121123303\allowbreak{}3200231314223111\allowbreak{}0414100340414443\allowbreak{}3410304030224041\allowbreak{}1}\newline \(r_2\) (degree $28$): \coef{4302022311043200\allowbreak{}0403341310440}\\[4pt]
N13 & \(g\) (degree $65$): \coef{1221442230200413\allowbreak{}2242440432303422\allowbreak{}2302344444032001\allowbreak{}2320204433224100\allowbreak{}04}\newline \(r_2\) (degree $27$): \coef{2204012023130002\allowbreak{}412213432034}\\[4pt]
N14 & \(g\) (degree $68$): \coef{1113011413444004\allowbreak{}0422014112223440\allowbreak{}0143341114433413\allowbreak{}0314130124230231\allowbreak{}41301}\newline \(r_2\) (degree $23$): \coef{3302042333214201\allowbreak{}21112420}\\[4pt]
N15 & \(g\) (degree $70$): \coef{1312421304021434\allowbreak{}4443214432433211\allowbreak{}0012221440121423\allowbreak{}4313401213234140\allowbreak{}0404241}\newline \(r_2\) (degree $22$): \coef{3033430040412322\allowbreak{}4033220}\\[4pt]
N16 & \(g\) (degree $71$): \coef{1340124403200141\allowbreak{}3133022012130040\allowbreak{}4204432233124312\allowbreak{}2403332122113314\allowbreak{}12104204}\newline \(r_2\) (degree $21$): \coef{2234200403411404\allowbreak{}031041}\\[4pt]
N17 & \(g\) (degree $72$): \coef{1334021323222020\allowbreak{}1021130000221430\allowbreak{}3023022033211313\allowbreak{}4242122340323222\allowbreak{}441224321}\newline \(r_2\) (degree $20$): \coef{1434221114023403\allowbreak{}14421}\\[4pt]
N18 & \(g\) (degree $27$): \coef{1001222101010112\allowbreak{}220001210102}\newline \(r_2\) (degree $24$): \coef{2022012211112000\allowbreak{}202212001}\\[4pt]
N19 & \(g\) (degree $24$): \coef{1344001212100243\allowbreak{}243214331}\newline \(r_2\) (degree $6$): \coef{3114034}\\[4pt]
N20 & \(g\) (degree $17$): \coef{1232212242234010\allowbreak{}41}\newline \(r_2\) (degree $24$): \coef{4434023012433223\allowbreak{}140003111}\\[4pt]
N21 & \(g\) (degree $8$): \coef{104540151}\newline \(r_2\) (degree $20$): \coef{5610163135340331\allowbreak{}43660}\\[4pt]
N22 & \(g\) (degree $9$): \coef{1515410636}\newline \(r_2\) (degree $22$): \coef{6004030645236406\allowbreak{}4553156}\\[4pt]
N23 & \(g\) (degree $7$): \coef{15025063}\newline \(r_2\) (degree $28$): \coef{1225261620301121\allowbreak{}5313204063123}\\[4pt]
N24 & \(g\) (degree $10$): \coef{16031050311}\newline \(r_2\) (degree $26$): \coef{4511153230530324\allowbreak{}64515033442}\\[4pt]
N25 & \(g\) (degree $11$): \coef{116423220331}\newline \(r_2\) (degree $26$): \coef{1603123540643545\allowbreak{}15064502013}\\[4pt]
N26 & \(g\) (degree $19$): \coef{1636362220615116\allowbreak{}3011}\newline \(r_2\) (degree $18$): \coef{4026316502526326\allowbreak{}624}\\[4pt]
N27 & \(g\) (degree $13$): \coef{14432026041451}\newline \(r_2\) (degree $26$): \coef{3540550016151156\allowbreak{}33312616162}\\[4pt]
N28 & \(g\) (degree $24$): \coef{1015502404416144\allowbreak{}042055101}\newline \(r_2\) (degree $18$): \coef{6564225210363115\allowbreak{}054}\\[4pt]
N29 & \(g\) (degree $34$): \coef{1154533521620044\allowbreak{}1403331603344413\allowbreak{}524}\newline \(r_2\) (degree $22$): \coef{2233613516121566\allowbreak{}3351052}\\[4pt]
N30 & \(g\) (degree $19$): \coef{1363541130606134\allowbreak{}0262}\newline \(r_2\) (degree $28$): \coef{3260354026363236\allowbreak{}6552500030502}\\[4pt]
N31 & \(g\) (degree $25$): \coef{1002053531443356\allowbreak{}2460652016}\newline \(r_2\) (degree $29$): \coef{5116354300605411\allowbreak{}16100234514146}\\[4pt]
N32 & \(g\) (degree $26$): \coef{1462354043522033\allowbreak{}42434340364}\newline \(r_2\) (degree $30$): \coef{2153340453364112\allowbreak{}356014003144002}\\[4pt]
N33 & \(g\) (degree $27$): \coef{1565611505132416\allowbreak{}200065463546}\newline \(r_2\) (degree $29$): \coef{3660162015043466\allowbreak{}41153626015036}\\[4pt]
N34 & \(g\) (degree $25$): \coef{1234655316346404\allowbreak{}2253224626}\newline \(r_2\) (degree $31$): \coef{2546363263000155\allowbreak{}1204656104605051}\\[4pt]
N35 & \(g\) (degree $11$): \coef{132510165645}\newline \(r_2\) (degree $36$): \coef{6105455546302156\allowbreak{}5156400361000460\allowbreak{}25166}\\[4pt]
N36 & \(g\) (degree $20$): \coef{1002063551301162\allowbreak{}25061}\newline \(r_2\) (degree $34$): \coef{5403642516135351\allowbreak{}6204510453510661\allowbreak{}201}\\[4pt]
N37 & \(g\) (degree $6$): \coef{1341241}\newline \(r_2\) (degree $41$): \coef{4142400201040332\allowbreak{}2333432212010001\allowbreak{}3213102104}\\[4pt]
N38 & \(g\) (degree $9$): \coef{1641432243}\newline \(r_2\) (degree $38$): \coef{5334120515461240\allowbreak{}1456134661332005\allowbreak{}2001321}\\[4pt]
N39 & \(g\) (degree $9$): \coef{1635445361}\newline \(r_2\) (degree $40$): \coef{1241560444421263\allowbreak{}2102652663334315\allowbreak{}142605050}\\[4pt]
N40 & \(g\) (degree $10$): \coef{13022110222}\newline \(r_2\) (degree $41$): \coef{4341410041024130\allowbreak{}1132314313400400\allowbreak{}2232441414}\\[4pt]
N41 & \(g\) (degree $3$): \coef{1122}\newline \(r_2\) (degree $44$): \coef{1211034444122210\allowbreak{}1121311110222204\allowbreak{}4403010210100}\\[4pt]
N42 & \(g\) (degree $6$): \coef{1130554}\newline \(r_2\) (degree $41$): \coef{5055643312302005\allowbreak{}0410613002420102\allowbreak{}5142233661}\\[4pt]
N43 & \(g\) (degree $19$): \coef{1260540060126633\allowbreak{}1353}\newline \(r_2\) (degree $37$): \coef{4226163616422150\allowbreak{}0035226526466661\allowbreak{}552516}\\[4pt]
N44 & \(g\) (degree $8$): \coef{155040551}\newline \(r_2\) (degree $41$): \coef{3405254663654500\allowbreak{}6243065340433323\allowbreak{}5541124230}\\[4pt]
\end{longtable}
\endgroup
\section{\texorpdfstring{Generators for the matrix-twist codes}{Generators for the matrix-twist codes}}
For H1--H13, $m_t=x^{2n}-tx^n+1$ and $g=m_t/h$; a generator matrix is
given by Lemma~\ref{matrix-ideal-basis}. Table~\ref{tab:qt-direct}
supplies $q$, $n$, and $t$.

\begingroup
\small
\setlength{\tabcolsep}{5pt}
\renewcommand{\arraystretch}{1.0}
\begin{longtable}{@{}p{0.05\textwidth}>{\raggedright\arraybackslash}p{0.89\textwidth}@{}}
\caption{Generator coefficients for H1--H13.}\label{tab:qt-generators}\\
\toprule
No. & Polynomials\\
\midrule
\endfirsthead
\multicolumn{2}{c}{\tablename~\thetable\ (continued)}\\
\toprule
No. & Polynomials\\
\midrule
\endhead
\midrule
\multicolumn{2}{r}{Continued on the next page}\\
\endfoot
\bottomrule
\endlastfoot
H1 & \(h\) (degree $18$): \coef{1222101022010111\allowbreak{}101}\newline \(g\) (degree $164$): \coef{1121020010122122\allowbreak{}2011121111201012\allowbreak{}0202001121200011\allowbreak{}1200012021101211\allowbreak{}2010110112010210\allowbreak{}0122210220110011\allowbreak{}1022201012012011\allowbreak{}1001120200211100\allowbreak{}1201221121110101\allowbreak{}1000200210121021\allowbreak{}02201}\\[4pt]
H2 & \(h\) (degree $20$): \coef{1023242012020342\allowbreak{}33214}\newline \(g\) (degree $84$): \coef{1032232230031221\allowbreak{}2033334122020342\allowbreak{}0044131410223301\allowbreak{}3021300324123401\allowbreak{}1110210004230410\allowbreak{}02244}\\[4pt]
H3 & \(h\) (degree $22$): \coef{1042236042231511\allowbreak{}5415356}\newline \(g\) (degree $78$): \coef{1035066644505552\allowbreak{}0453444314640123\allowbreak{}4604243435242061\allowbreak{}4340231554421062\allowbreak{}124325516101026}\\[4pt]
H4 & \(h\) (degree $36$): \coef{1013030020031111\allowbreak{}1332141412002002\allowbreak{}02101}\newline \(g\) (degree $48$): \coef{1042133201231200\allowbreak{}3302344441413302\allowbreak{}3003122403321340\allowbreak{}1}\\[4pt]
H5 & \(h\) (degree $30$): \coef{1523543442534644\allowbreak{}450011450150046}\newline \(g\) (degree $70$): \coef{1224002162264311\allowbreak{}0331522011154245\allowbreak{}3246161142534602\allowbreak{}5143424141040100\allowbreak{}3656536}\\[4pt]
H6 & \(h\) (degree $30$): \coef{1001202210200202\allowbreak{}122002021000211}\newline \(g\) (degree $152$): \coef{1002102200111012\allowbreak{}2201101122012221\allowbreak{}2202012010212100\allowbreak{}1000022222020202\allowbreak{}0201220200222020\allowbreak{}2010000201012012\allowbreak{}0121021122022200\allowbreak{}1101101220111002\allowbreak{}1011020011002000\allowbreak{}000120221}\\[4pt]
H7 & \(h\) (degree $32$): \coef{1430433322443431\allowbreak{}4224410202431331\allowbreak{}1}\newline \(g\) (degree $72$): \coef{1130201114400221\allowbreak{}3033301432031321\allowbreak{}0300044440314211\allowbreak{}4121334101244343\allowbreak{}020212341}\\[4pt]
H8 & \(h\) (degree $36$): \coef{1002004214403240\allowbreak{}2031244323240122\allowbreak{}02401}\newline \(g\) (degree $68$): \coef{1003000344442110\allowbreak{}0000034132111230\allowbreak{}3340023220213343\allowbreak{}0122441004402434\allowbreak{}13101}\\[4pt]
H9 & \(h\) (degree $38$): \coef{1232204202012103\allowbreak{}0330033440100420\allowbreak{}2133044}\newline \(g\) (degree $66$): \coef{1312013330433134\allowbreak{}1100001124013231\allowbreak{}2131320032022203\allowbreak{}4334223433333223\allowbreak{}414}\\[4pt]
H10 & \(h\) (degree $30$): \coef{1605316415614321\allowbreak{}022426160235321}\newline \(g\) (degree $120$): \coef{1113203126512122\allowbreak{}3602436043164664\allowbreak{}4101004104306105\allowbreak{}4401322553015662\allowbreak{}1312201525466522\allowbreak{}5350143051253066\allowbreak{}1103605313104153\allowbreak{}001466151}\\[4pt]
H11 & \(h\) (degree $38$): \coef{1353562144625516\allowbreak{}5503462644203551\allowbreak{}0015336}\newline \(g\) (degree $62$): \coef{1440534214300404\allowbreak{}6610420216641214\allowbreak{}0266151251634056\allowbreak{}552334230016246}\\[4pt]
H12 & \(h\) (degree $40$): \coef{1320123431431023\allowbreak{}2310224303133441\allowbreak{}001134314}\newline \(g\) (degree $64$): \coef{1220013023024442\allowbreak{}3022230433200141\allowbreak{}0324113423321033\allowbreak{}4411131322420414\allowbreak{}4}\\[4pt]
H13 & \(h\) (degree $40$): \coef{1022140123362251\allowbreak{}4461611031504246\allowbreak{}164415501}\newline \(g\) (degree $60$): \coef{1055340220024141\allowbreak{}4465335655243666\allowbreak{}6420060560310324\allowbreak{}4614354432201}\\[4pt]
\end{longtable}
\endgroup

\end{document}